\documentclass[journal, twocolumn]{IEEEtran}

\ifCLASSINFOpdf
\else
\fi

\usepackage[cmex10]{amsmath}

\usepackage{amssymb,amsthm}
\usepackage{color}
\usepackage{graphics}

\usepackage{multirow}
\usepackage{hhline}

\usepackage{tikz}
\usetikzlibrary{angles,calc,intersections,quotes,arrows.meta}
\usetikzlibrary{calc, arrows,backgrounds,positioning,fit}
\usetikzlibrary{decorations.markings}
\tikzset{middlearrow/.style={
        decoration={markings,
            mark= at position 0.5 with {\arrow{#1}} ,
        },
        postaction={decorate}
    }
}

\allowdisplaybreaks

\usepackage[noadjust]{cite}

\makeatletter
\newtheorem*{rep@theorem}{\rep@title}
\newcommand{\newreptheorem}[2]{%
\newenvironment{rep#1}[1]{%
 \def\rep@title{#2 \ref{##1}}%
 \begin{rep@theorem}}%
 {\end{rep@theorem}}}
\makeatother

\theoremstyle{definition}
\newtheorem{definition}{Definition}

\theoremstyle{plain}
\newtheorem{theorem}{Theorem}
\newreptheorem{theorem}{Theorem}
\newtheorem{proposition}{Proposition}
\newtheorem{lemma}{Lemma}
\newtheorem{remark}{Remark}

\DeclareMathOperator{\Rk}{Rk}
\DeclareMathOperator{\diag}{diag}
\DeclareMathOperator{\Row}{Row}
\DeclareMathOperator{\Col}{Col}
\DeclareMathOperator{\dd}{d}
\DeclareMathOperator{\wt}{wt}
\DeclareMathOperator{\Id}{Id}

\begin{document}
%
\title{Reliable and Secure Multishot Network Coding using Linearized Reed-Solomon Codes}
%
%
%
\author{Umberto~Mart{\'i}nez-Pe\~{n}as,~\IEEEmembership{Member,~IEEE,} 
		and Frank~R.~Kschischang,~\IEEEmembership{Fellow,~IEEE,}
\thanks{This work is supported by The Independent Research Fund Denmark under Grant No. DFF-7027-00053B. }
\thanks{Parts of this paper were presented at the 56th Annual Allerton Conference on Communication, Control, and Computing, Monticello, IL, USA, 2018.}
\thanks{U. Mart{\'i}nez-Pe\~{n}as and F.~R.~Kschischang are with The Edward
S. Rogers Sr. Department of Electrical and Computer Engineering, University
of Toronto, Toronto, ON M5S 3G4, Canada. (e-mail: umberto@ece.utoronto.ca;
frank@ece.utoronto.ca)}}


%
%

\markboth{}%
{}
%



\maketitle



\begin{abstract}
Multishot network coding is considered in a worst-case adversarial
setting in which an omniscient adversary with unbounded computational
resources may inject erroneous packets in up to $t$ links, erase up to
$\rho$ packets, and wire-tap up to $\mu$ links, all throughout $\ell$
shots of a linearly-coded network.  Assuming no knowledge of
the underlying linear network code (in particular, the 
network topology and underlying linear code may be random and change with time), a 
coding scheme achieving zero-error
communication and perfect secrecy is obtained based on linearized
Reed-Solomon codes. The scheme achieves the maximum possible secret
message size of $ \ell n^\prime
- 2t - \rho - \mu $ packets for coherent communication, where $ n^\prime $ is the number of
  outgoing links at the source, for any packet length $ m \geq n^\prime $ (largest
possible range). By lifting this construction, coding schemes for
non-coherent communication are obtained with information rates close to
optimal for practical instances. The required field size is $ q^m $, where
$ q > \ell $, thus $ q^m \approx \ell^{n^\prime} $, which is always smaller than that of
a Gabidulin code tailored for $ \ell $ shots, which would be at least $ 2^{\ell n^\prime} $.  A 
Welch-Berlekamp sum-rank decoding
algorithm for linearized Reed-Solomon codes is provided, having
quadratic complexity in the total length $n = \ell n^\prime $, and which
can be adapted to handle not only errors, but also erasures, wire-tap
observations and non-coherent communication. Combined with the 
obtained field size, the given decoding complexity is of 
$ \mathcal{O}(n^{\prime 4} \ell^2 \log(\ell)^2) $ operations in $ \mathbb{F}_2 $, whereas
the most efficient known decoding algorithm for a Gabidulin code has
a complexity of $ \mathcal{O}(n^{\prime 3.69} \ell^{3.69} \log(\ell)^2) $ operations in $ \mathbb{F}_2 $, assuming a multiplication in a finite field $ \mathbb{F} $ costs about $ \log(|\mathbb{F}|)^2 $ operations in $ \mathbb{F}_2 $.
\end{abstract}


%
\IEEEpeerreviewmaketitle

\section{Introduction} \label{sec intro}
%
%
%
%
\IEEEPARstart{L}{inear} network coding over a finite field $ \mathbb{F}_{q_0} $, introduced in \cite{ahlswede, linearnetwork, Koetter2003}, permits maximum information flow from a
source to several sinks simultaneously in one shot (\textit{multicast}).
Moreover, for sufficiently large field size $ q_0 $, the maximum information
flow can be achieved with high probability by a random choice of coding
coefficients at each node, without knowledge of the network topology
\cite{random}. 

Correction of link errors was considered in \cite{cai-yeung,
yeung-partI, yeung-partII, jaggi-resilient, zhang-network,
yang-refined}, and secrecy against link wire-tapping was studied in
\cite{secure-network, cai-secure-network, feldman, jaggi-isit,
wiretapnetworks}. Some of these works assume probabilistic error
correction, and some require knowledge or modification of the linear
network code for the outer code construction.  Error-correcting codes
under an adversarial model without such requirements (thus compatible with random linear network coding) were first given in
\cite{errors-network, error-control} for \textit{non-coherent
communication} (in which the sink has no knowledge of the coding
coefficients of the incoming links), and in \cite{on-metrics} for
\textit{coherent communication}. Coding schemes that provide perfect
secrecy and zero-error communication, without knowledge or modification of the underlying linear network code, were first given in \cite{silva-universal}.
The coherent-case construction in that work (similarly in
\cite{errors-network, error-control, on-metrics}) is based on Gabidulin
codes \cite{gabidulin, roth}. If the \textit{packets} sent through the links
of the network in one shot are vectors in $ \mathbb{F}_{q_0}^m $, then coding schemes
based on Gabidulin codes achieve the maximum secret message
size of $ n^\prime - 2t - \rho - \mu $ packets, where $ n^\prime $ is
the number of outgoing links at the source, and for $ t $ (link) errors, $ \rho $
erasures and $ \mu $ wire-tapped links. Moreover, the packet length $ m $ is
only restricted to $ m \geq n^\prime $, which is the maximum possible
range for $ m $, where $ n^\prime $ is a constraint given by the channel.

\begin{figure*} [!t]
\begin{center}
\begin{tabular}{c@{\extracolsep{1cm}}c}
	\begin{tikzpicture}[line width=1pt, scale=0.77]
		\tikzstyle{every node}=[inner sep=0pt, minimum width=4.5pt]
		\draw (-3,0) node (v1) [draw, circle, fill=gray] {};
		\draw (-1.5,2) node (v2) [draw, circle, fill=gray] {};
		\draw (-1.5,-2) node (v3) [draw, circle, fill=gray] {};
		\draw (0,0) node (v4) [draw, circle, fill=gray] {};
		\draw (1.5,0) node (v5) [draw, circle, fill=gray] {};
		\draw (3,2) node (v6) [draw, circle, fill=gray] {};
		\draw (3,-2) node (v7) [draw, circle, fill=gray] {};
		\draw (-3.6,0) node () {\small $S$};
		\draw (3.6,2) node () {\small $T_1$};
		\draw (3.6,-2) node () {\small $T_2$};
		\draw[middlearrow={stealth}, ultra thick] (v1) -- (v2);
		\draw[middlearrow={stealth}, ultra thick] (v2) -- (v6);
		\draw[middlearrow={stealth}, ultra thick] (v5) -- (v6);
		\draw[middlearrow={stealth}, ultra thick] (v4) -- (v5);
		\draw[middlearrow={stealth}, ultra thick] (v2) -- (v4);
		\draw[middlearrow={stealth}, ultra thick] (v1) -- (v3);
		\draw[red, dashed] (v3) -- (v7);
		\draw[middlearrow={stealth}, ultra thick] (v5) -- (v7);
		\draw[middlearrow={stealth}, ultra thick] (v3) -- (v4);
		\draw (-2.7,1.2) node () {\small $a_1$};
		\draw (-2.7,-1.2) node () {\small $b_1$};
		\draw (-1.2,0.7) node () {\small $a_1$};
		\draw (-1.2,-0.7) node () {\small $b_1$};
		\draw (0.8,2.5) node () {\small $a_1$};
		\draw (0.8,-2.5) node () {\small $b_1$};
		\draw (0.8,0.6) node () {\small $a_1+b_1$};
		\draw (1.5,1.2) node () {\small $a_1+b_1$};
		\draw (1.5,-1.2) node () {\small $a_1+b_1$};
	\end{tikzpicture}
	
	\begin{tikzpicture}[line width=1pt, scale=0.77]
		\tikzstyle{every node}=[inner sep=0pt, minimum width=4.5pt]
		\draw (-3,0) node (v1) [draw, circle, fill=gray] {};
		\draw (-1.5,2) node (v2) [draw, circle, fill=gray] {};
		\draw (-1.5,-2) node (v3) [draw, circle, fill=gray] {};
		\draw (0,0) node (v4) [draw, circle, fill=gray] {};
		\draw (1.5,0) node (v5) [draw, circle, fill=gray] {};
		\draw (3,2) node (v6) [draw, circle, fill=gray] {};
		\draw (3,-2) node (v7) [draw, circle, fill=gray] {};
		\draw (-3.6,0) node () {\small $S$};
		\draw (3.6,2) node () {\small $T_1$};
		\draw (3.6,-2) node () {\small $T_2$};
		\draw[middlearrow={stealth}, ultra thick] (v1) -- (v2);
		\draw[middlearrow={stealth}, ultra thick] (v2) -- (v6);
		\draw[middlearrow={stealth}, ultra thick] (v5) -- (v6);
		\draw[middlearrow={stealth}, ultra thick] (v4) -- (v5);
		\draw[middlearrow={stealth}, ultra thick] (v2) -- (v4);
		\draw[middlearrow={stealth}, ultra thick] (v1) -- (v3);
		\draw[middlearrow={stealth}, ultra thick] (v3) -- (v7);
		\draw[middlearrow={stealth}, ultra thick] (v5) -- (v7);
		\draw[middlearrow={stealth}, ultra thick] (v3) -- (v4);
		\draw (-2.7,1.2) node () {\small $a_2$};
		\draw (-2.7,-1.2) node () {\small $b_2$};
		\draw (-1.2,0.7) node () {\small $a_2$};
		\draw (-1.2,-0.7) node () {\small $b_2$};
		\draw (0.8,2.5) node () {\small $a_2$};
		\draw (0.8,-2.5) node () {\small $b_2$};
		\draw (0.8,0.6) node () {\small $a_2$};
		\draw (1.5,1.2) node () {\small $a_2$};
		\draw (1.5,-1.2) node () {\small $a_2$};
	\end{tikzpicture}
	
	\begin{tikzpicture}[line width=1pt, scale=0.77]
		\tikzstyle{every node}=[inner sep=0pt, minimum width=4.5pt]
		\draw (-3,0) node (v1) [draw, circle, fill=gray] {};
		\draw (-1.5,2) node (v2) [draw, circle, fill=gray] {};
		\draw (-1.5,-2) node (v3) [draw, circle, fill=gray] {};
		\draw (0,0) node (v4) [draw, circle, fill=gray] {};
		\draw (1.5,0) node (v5) [draw, circle, fill=gray] {};
		\draw (3,2) node (v6) [draw, circle, fill=gray] {};
		\draw (3,-2) node (v7) [draw, circle, fill=gray] {};
		\draw (-3.6,0) node () {\small $S$};
		\draw (3.6,2) node () {\small $T_1$};
		\draw (3.6,-2) node () {\small $T_2$};
		\draw[middlearrow={stealth}, ultra thick] (v1) -- (v2);
		\draw[middlearrow={stealth}, ultra thick] (v2) -- (v6);
		\draw[red, dotted] (v5) -- (v6);
		\draw[red, dashed] (v4) -- (v5);
		\draw[middlearrow={stealth}, ultra thick] (v2) -- (v4);
		\draw[middlearrow={stealth}, ultra thick] (v1) -- (v3);
		\draw[red, dashed] (v3) -- (v7);
		\draw[red, dotted] (v5) -- (v7);
		\draw[middlearrow={stealth}, ultra thick] (v3) -- (v4);
		\draw (-2.7,1.2) node () {\small $a_3$};
		\draw (-2.7,-1.2) node () {\small $b_3$};
		\draw (-1.2,0.7) node () {\small $a_3$};
		\draw (-1.2,-0.7) node () {\small $b_3$};
		\draw (0.8,2.5) node () {\small $a_3$};
		\draw (0.8,-2.5) node () {\small $b_3$};
		\draw (0.8,0.6) node () {\small $a_3+b_3$};
		\draw (1.5,1.2) node () {\small $a_3+b_3$};
		\draw (1.5,-1.2) node () {\small $a_3+b_3$};
	\end{tikzpicture}
\end{tabular}
\end{center}

\caption{A pattern of link errors in multishot linear network coding.
Similarly for link observations by a wire-tapper. Dashed lines denote
links directly affected by the adversary, and dotted lines denote links
that suffer the effect indirectly. The network code may change with time. In this example, a rank deficiency results in an erasure for the first sink in the second shot. Although not depicted here, the network topology may also change with time. If the linear network code is defined over $ \mathbb{F}_{q_0} = \mathbb{F}_2 $, linearized Reed-Solomon codes in this case are defined over $ \mathbb{F}_{q^{n^\prime}} = \mathbb{F}_{q^2} $, $ 2 | q $, and allow us to use $ \ell < q $ shots of the network, thus $ | \mathbb{F}_{q^2} | \approx \ell^2 $, which is quadratic in $ \ell $. Observe that Gabidulin codes would require the field $ \mathbb{F}_{2^{2 \ell}} $, which is exponential in $ \ell $.}
\label{fig multishot network}
\end{figure*}
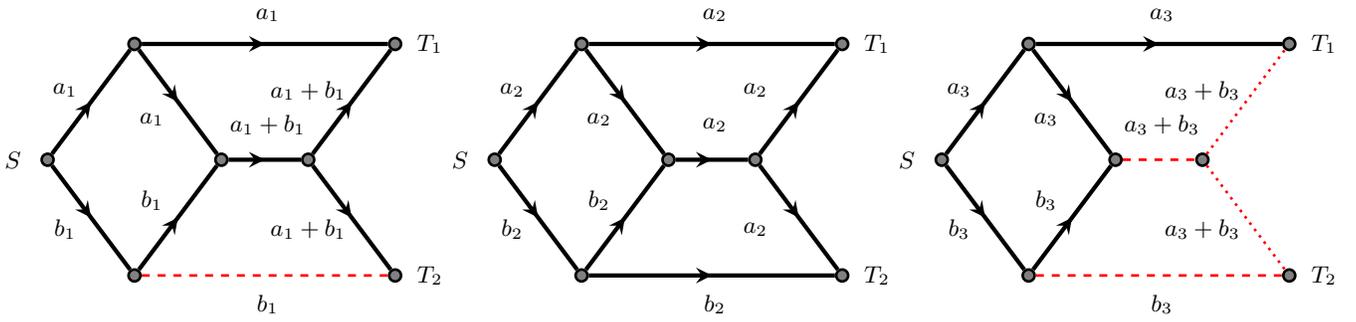

\begin{figure*} [!t]
\begin{center}
\begin{tabular}{c@{\extracolsep{1cm}}c}
	\begin{tikzpicture}[line width=1pt, scale=0.77]
		\tikzstyle{every node}=[inner sep=0pt, minimum width=4.5pt]
		\draw (-3,0) node (v1) [draw, circle, fill=gray] {};
		\draw (-1.5,0) node (v2) [draw, circle, fill=gray] {};
		\draw (0,-1) node (v3) [draw, circle, fill=gray] {};
		\draw (0,1) node (v4) [draw, circle, fill=gray] {};
		\draw (1.5,-1) node (v45) [draw, circle, fill=gray] {};		
		\draw (1.5,1) node (v5) [draw, circle, fill=gray] {};
		\draw (3,1) node (v6) [draw, circle, fill=gray] {};
		\draw (3,-1) node (v7) [draw, circle, fill=gray] {};
		\draw (-3.6,0) node () {\small $S$};
		\draw (3.6,1) node () {\small $T_1$};
		\draw (3.6,-1) node () {\small $T_2$};
		\draw[middlearrow={stealth}, ultra thick] (v1) -- (v2);
		\draw[middlearrow={stealth}, ultra thick] (v2) -- (v3);
		\draw[middlearrow={stealth}, ultra thick] (v2) -- (v4);
		\draw[middlearrow={stealth}, ultra thick] (v3) -- (v45);
		\draw[middlearrow={stealth}, ultra thick] (v4) -- (v5);			
		\draw[middlearrow={stealth}, ultra thick] (v45) -- (v7);
		\draw[middlearrow={stealth}, ultra thick] (v5) -- (v6);	
		\draw (-2.3,0.5) node () {\small $a_1$};
		\draw (-1,1) node () {\small $a_1$};
		\draw (-1,-1) node () {\small $a_1$};
		\draw (0.8,0.5) node () {\small $a_1$};
		\draw (0.8,-0.5) node () {\small $a_1$};
		\draw (2.3,0.5) node () {\small $a_1$};
		\draw (2.3,-0.5) node () {\small $a_1$};
	\end{tikzpicture}
	
	\begin{tikzpicture}[line width=1pt, scale=0.77]
		\tikzstyle{every node}=[inner sep=0pt, minimum width=4.5pt]
		\draw (-3,0) node (v1) [draw, circle, fill=gray] {};
		\draw (-1.5,0) node (v2) [draw, circle, fill=gray] {};
		\draw (0,-1) node (v3) [draw, circle, fill=gray] {};
		\draw (0,1) node (v4) [draw, circle, fill=gray] {};
		\draw (1.5,-1) node (v45) [draw, circle, fill=gray] {};		
		\draw (1.5,1) node (v5) [draw, circle, fill=gray] {};
		\draw (3,1) node (v6) [draw, circle, fill=gray] {};
		\draw (3,-1) node (v7) [draw, circle, fill=gray] {};
		\draw (-3.6,0) node () {\small $S$};
		\draw (3.6,1) node () {\small $T_1$};
		\draw (3.6,-1) node () {\small $T_2$};
		\draw[middlearrow={stealth}, ultra thick] (v1) -- (v2);
		\draw[middlearrow={stealth}, ultra thick] (v2) -- (v3);
		\draw[middlearrow={stealth}, ultra thick] (v2) -- (v4);
		\draw[red, dashed] (v3) -- (v45);
		\draw[middlearrow={stealth}, ultra thick] (v4) -- (v5);			
		\draw[red, dotted] (v45) -- (v7);
		\draw[middlearrow={stealth}, ultra thick] (v5) -- (v6);	
		\draw (-2.3,0.5) node () {\small $a_2$};
		\draw (-1,1) node () {\small $a_2$};
		\draw (-1,-1) node () {\small $a_2$};
		\draw (0.8,0.5) node () {\small $a_2$};
		\draw (0.8,-0.5) node () {\small $a_2$};
		\draw (2.3,0.5) node () {\small $a_2$};
		\draw (2.3,-0.5) node () {\small $a_2$};
	\end{tikzpicture}
	
	\begin{tikzpicture}[line width=1pt, scale=0.77]
		\tikzstyle{every node}=[inner sep=0pt, minimum width=4.5pt]
		\draw (-3,0) node (v1) [draw, circle, fill=gray] {};
		\draw (-1.5,0) node (v2) [draw, circle, fill=gray] {};
		\draw (0,-1) node (v3) [draw, circle, fill=gray] {};
		\draw (0,1) node (v4) [draw, circle, fill=gray] {};
		\draw (1.5,-1) node (v45) [draw, circle, fill=gray] {};		
		\draw (1.5,1) node (v5) [draw, circle, fill=gray] {};
		\draw (3,1) node (v6) [draw, circle, fill=gray] {};
		\draw (3,-1) node (v7) [draw, circle, fill=gray] {};
		\draw (-3.6,0) node () {\small $S$};
		\draw (3.6,1) node () {\small $T_1$};
		\draw (3.6,-1) node () {\small $T_2$};
		\draw[middlearrow={stealth}, ultra thick] (v1) -- (v2);
		\draw[middlearrow={stealth}, ultra thick] (v2) -- (v3);
		\draw[middlearrow={stealth}, ultra thick] (v2) -- (v4);
		\draw[middlearrow={stealth}, ultra thick] (v3) -- (v45);
		\draw[red, dashed] (v4) -- (v5);			
		\draw[middlearrow={stealth}, ultra thick] (v45) -- (v7);
		\draw[red, dotted] (v5) -- (v6);	
		\draw (-2.3,0.5) node () {\small $a_3$};
		\draw (-1,1) node () {\small $a_3$};
		\draw (-1,-1) node () {\small $a_3$};
		\draw (0.8,0.5) node () {\small $a_3$};
		\draw (0.8,-0.5) node () {\small $a_3$};
		\draw (2.3,0.5) node () {\small $a_3$};
		\draw (2.3,-0.5) node () {\small $a_3$};
	\end{tikzpicture}
\end{tabular}
\end{center}

\caption{Reliable and secure multishot linear network coding naturally extends reliable and secure multicast discrete memoryless channels. In this figure, the latter is depicted. Network coding is not necessary, which means that we may take the network base field $ \mathbb{F}_{q_0} = \mathbb{F}_2 $. Since $ n^\prime = 1 $, linearized Reed-Solomon codes coincide in this case with classical and generalized Reed-Solomon codes \cite{reed-solomon, delsarte}, which are defined over $ \mathbb{F}_q $ and allow us to use $ \ell < q $ shots, thus $ | \mathbb{F}_q | \approx \ell $, which is linear in $ \ell $. In this case, Gabidulin codes would require the field $ \mathbb{F}_{2^\ell} $, which is again exponential in $ \ell $ .}
\label{fig multishot network DMC}
\end{figure*}
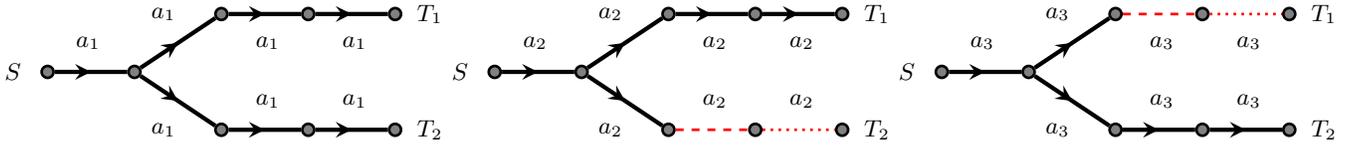

All of the works noted above make use of only \textit{one shot} of the
linearly-coded network. Correction of link errors in \textit{multishot
network coding} (see Fig.~\ref{fig multishot network}) was first
investigated in \cite{multishot2009, multishot}.  As noted in these
works,  the $ \ell $-shot case can be treated as a $ 1 $-shot case with
number of outgoing links at the source $ n = \ell n^\prime $. An $ \ell $-shot
Gabidulin code, that is, a Gabidulin code tailored for $ \ell n^\prime $ outgoing 
links at the source, yields the maximum message size of $ \ell n^\prime - 2t - \rho $
packets, but would require packet lengths $ m \geq n = \ell n^\prime $
instead of $ m \geq n^\prime $, which may be
impractically large. More importantly, decoding an $ \ell $-shot Gabidulin code using \cite{decoding-loidreau}
would require $ \mathcal{O}(n^2) $ operations over a field of size
$ q_0^{\ell n^\prime} $ (even faster decoders \cite{subquadratic,
silva-fast} would be quite expensive for such field sizes,
see Table \ref{table comparison}). To
circumvent these issues, the authors of \cite{multishot2009, multishot}
provide a \textit{multilevel} construction that improves the error-correction
capability of trivial
concatenations of $ 1 $-shot Gabidulin codes, with lower decoding complexity
than an $ \ell $-shot Gabidulin code, although without achieving the maximum secret message size. 

Later, \textit{convolutional} rank-metric codes were studied in
\cite{wachter, wachter-convolutional, mahmood-convolutional,
ashish-streaming, mrd-convolutional}, and a concatenation of an outer
Hamming-metric convolutional code with an inner rank-metric block code
was given in \cite{napp-sidorenko}. See \cite{cost-action-multishot} for
a survey. Convolutional techniques yield codes achieving the streaming
capacity for the burst rank loss networks described in
\cite{ashish-streaming}, where recovery of a given packet is required
before a certain delay. By truncating the convolutional codes in 
\cite{mahmood-convolutional, ashish-streaming}
up to their memory, one may obtain block codes whose message size is close to
the upper bound $ \ell n^\prime - 2t - \rho $, but their field sizes are in general
still exponential in $ n^\prime $ and $ \ell $. 

To the best of our knowledge, schemes that achieve zero-error communication and perfect secrecy
without knowledge of the underlying network code as in
\cite{silva-universal} have not yet been investigated in the multishot
case.

In this work, we provide such a family of coding schemes for
coherent communication with maximum secret message size of $ \ell
n^\prime - 2t - \rho - \mu $ packets, whose packet length is only
restricted to $ m \geq n^\prime $ (maximum possible range) in contrast
to $ m \geq \ell n^\prime $. The original base field of the 
network is embedded $ \mathbb{F}_{q_0} \subseteq \mathbb{F}_q $ in a 
field of size $ q = q_0^s $, for suitable $ s $ satisfying $ q > \ell $, 
and the obtained coding schemes are defined over the field $ \mathbb{F}_{q^m} $,
whose size is approximately $ \max \{ \ell, q_0 \}^{n^\prime} $, and in most real scenarios ($ \ell > q_0 $), it is $ \ell^{n^\prime} $. We do not know
if the range of values $ q > \ell $, for $ q $ given $ \ell $, is
maximum, but we conjecture that $ q $ must be $ \Omega(\ell) $ (equivalently, $ \ell = \mathcal{O}(q) $) as is the case for MDS codes (the case $ n^\prime = 1 $). By lifting our
construction, we adapt it to non-coherent communication with nearly
optimal information rates for practical instances. 

Our coding schemes are based on \textit{linearized Reed-Solomon codes},
introduced in \cite{linearizedRS} and closely connected to \textit{skew
Reed-Solomon codes} \cite{skew-evaluation1}. Linearized Reed-Solomon
codes are \textit{maximum sum-rank distance} (MSRD) codes, which is precisely the property that gives the optimality of our coding schemes. We provide a Welch-Berlekamp
sum-rank decoding algorithm for these codes that requires $ \mathcal{O}(n^2) $ operations
over the field of size $ q^m \approx \ell^{n^\prime} $. Hence we obtain a reduction in the decoding complexity
in operations over $ \mathbb{F}_2 $ of more than a degree in the number of shots compared to the
existing decoding algorithms and codes. See Table \ref{table comparison}. We remark that our algorithm is related to a skew-metric
decoding algorithm for skew Reed-Solomon codes recently given in
\cite{boucher-decoding}. As we will see, such an algorithm can be
translated to a sum-rank decoding algorithm for linearized Reed-Solomon
codes, although this is not done in \cite{boucher-decoding}. In addition, the algorithm in \cite{boucher-decoding} has cubic
complexity over the field of size $ q^m \approx \ell^{n^\prime} $ and does not handle erasures, wire-tap observations and non-coherent communication, in contrast to our algorithm. 

Linearized Reed-Solomon codes are natural hybrids between Reed-Solomon codes \cite{reed-solomon, delsarte}
and Gabidulin codes \cite{gabidulin, roth, new-construction}, in the
same way that the sum-rank metric is a hybrid between the Hamming and
rank metrics, and adversarial multishot network coding is a hybrid
between adversarial network coding and adversarial discrete memoryless
channels (e.g., wire-tap channel of type II \cite{ozarow} or secret
sharing \cite{shamir}). See and compare Figs. \ref{fig multishot network} and \ref{fig multishot network DMC}. Furthermore, our Welch-Berlekamp decoding algorithm includes the classical one
\cite{welch-berlekamp} and its rank-metric version
\cite{decoding-loidreau} as particular cases. What happens is that rank-metric and Hamming-metric codes are
two extremal particular cases of sum-rank-metric codes that correspond
to $ \ell = 1 $ and $ n^\prime = 1 $, respectively. With this point of
view, the use of linearized Reed-Solomon codes instead of Gabidulin
codes as MSRD codes extends the idea of using Reed-Solomon codes rather
than Gabidulin codes as MDS codes (see Figs. \ref{fig multishot network} and \ref{fig multishot network DMC}): in both cases, information rates are
equal but Gabidulin codes require a field size exponentiated to the $
\ell $th power, thus providing no advantage unless $ \ell = 1 $ (the
rank-metric case). See Table \ref{table skew ev codes}
in Subsection \ref{subsec skew metrics} for a summarized classification 
of Reed-Solomon-like evaluation codes with maximum distance for certain 
metrics. 

We conclude by giving two other
interpretations of reliable and secure $ \ell $-shot network coding.
First, $ \ell $-shot network coding can be thought of as $ 1 $-shot
network coding where we have knowledge of at least $ \ell $ connected
components in the underlying graph after removing the source and sinks.
Observe that here we may have a different number $ n_i $ of outgoing
links in each component. If each connected component has a single
outgoing link, we recover the (multicast) discrete memoryless channel (see Fig. \ref{fig multishot network DMC}).
Second, we may view $ \ell $-shot network coding as $ 1 $-shot network
coding where packets are divided into $ \ell $ sub-packets.  In this
scenario, errors, erasures and wire-tapper observations occur in certain
sub-packets (rather than over whole packets) in certain links. If the
sub-packets are considered as symbols, this gives another interpretation
of multishot network coding as hybrid between network coding and a
discrete memoryless channel.  The work \cite{hybrid-network} treats a
similar hybrid of symbol and link errors, although they consider that if
a symbol error occurs in one link, it occurs in all links (i.e., all
transmitted subspaces).  This may make error-correction impossible if
the adversary spreads the errors in pessimistic patterns.


\subsection*{Summary and Significance of our Main Contributions} \label{subsec main contribution}

Our coding schemes achieve the maximum secret message size of $ n^\prime \ell - 2t - \rho - \mu $ packets for coherent communication (Theorem \ref{th optimal for coherent comm}) and nearly optimal information rates for noncoherent communication (Theorem \ref{th near optimal coding schemes non-coherent}), as stated above, and satisfy the following:

1) Their field size is $ q^m $, where $ q > \ell $, $ q = q_0^s $ for some $ s \in \mathbb{N} $, and $ m \geq n^\prime $. Thus their field size is approximately
\begin{equation}
\max \{ \ell, q_0 \}^{n^\prime},
\label{eq field size}
\end{equation}
whereas an $ \ell $-shot Gabidulin code always requires the field size $ q_0^{\ell n^\prime} $, which is always considerably larger, since: a) $ \ell^{n^\prime} \ll (q_0^\ell)^{n^\prime} $, if $ q_0 < \ell $, and b) $ q_0^{n^\prime} \ll (q_0^{n^\prime})^\ell $, if $ q_0 \geq \ell $. Notice, moreover, that $ \ell^{n^\prime} $ is polynomial in $ \ell $, whereas $ q_0^{\ell n^\prime} $ is exponential in $ \ell $. See also Figs. \ref{fig multishot network} and \ref{fig multishot network DMC} for instances with $ n^\prime = 2 $ and $ n^\prime = 1 $, respectively, and $ q_0 = 2 $.

\begin{table}[h]
\caption{Decoding complexity in no. of operations over $ \mathbb{F}_2 $}
\label{table comparison}
\centering
\begin{tabular}{c|c|c}
\hline
&&\\[-0.8em]
Algorithm & Code & Complexity  \\[1.6pt]
\hline\hline
&&\\[-0.8em]
Subsec. \ref{subsec quadratic reduction} & Linearized RS \cite{linearizedRS} & $ \mathcal{O}(n^{\prime 4} \ell^2 \log(\ell)^2 ) $ \\[1.6pt]
\hline
&&\\[-0.8em]
\cite{boucher-decoding} & Skew RS \cite{skew-evaluation1} & $ \mathcal{O}(n^{\prime 5} \ell^3 \log(\ell)^2 ) $ \\[1.6pt]
\hline
&&\\[-0.8em]
\cite{subquadratic} & Gabidulin \cite{gabidulin} & $ \mathcal{O}(n^{\prime 3.69} \ell^{3.69} \log(\ell)^2) $ \\[1.6pt]
\hline
&&\\[-0.8em]
\cite{decoding-loidreau} & Gabidulin \cite{gabidulin} & $ \mathcal{O}(n^{\prime 4} \ell^{4}) $ \\[1.6pt]
\hline 
\end{tabular}
\end{table}

2) We provide a decoding algorithm with complexity of $ \mathcal{O}(n^2) $ operations over $ \mathbb{F}_{q^m} $ (recall that $ n = \ell n^\prime $). Assuming that one multiplication in a finite field $ \mathbb{F} $ of characteristic $ 2 $ costs about $ \log(|\mathbb{F}|)^2 $ operations in $ \mathbb{F}_2 $, our decoding algorithm has complexity of
\begin{equation}
\mathcal{O}(n^{\prime 4} \ell^2 \log(\ell)^2 )
\label{eq complexity over F_2}
\end{equation}
operations over $ \mathbb{F}_2 $, as a function of $ \ell $ and $ n^\prime $. Observe also that in most practical situations, $ \ell \gg n^\prime $ and $ n^\prime $ is a constant of the channel. Table \ref{table comparison} shows a comparison between this complexity and those of other algorithms in the literature, for coding schemes that are optimal in the coherent case or near optimal in the noncoherent case. We always assume that a multiplication in $ \mathbb{F} $ costs $ \log(|\mathbb{F}|)^2 $ operations in $ \mathbb{F}_2 $. Note that, setting $ n^\prime = 1 $, linearized Reed-Solomon codes become classical Reed-Solomon codes, and the obtained complexity in (\ref{eq complexity over F_2}) is of $ \mathcal{O}(\ell^2 \log(\ell)^2) $ operations in $ \mathbb{F}_2 $, which is that of the quadratic-complexity Welch-Berlekamp decoding algorithm.

3) By multiplying the message vector by a generator matrix, our coding scheme has encoding complexity of $ \mathcal{O}(kn) $ operations over $ \mathbb{F}_{q^m} $. Since encoding with a linear code has in general complexity of $ \mathcal{O}(kn) $ operations over the corresponding field, we obtain reductions in encoding complexity similar to those shown in Table \ref{table comparison}, in operations over $ \mathbb{F}_2 $.

\subsection*{Organization} \label{subsec organization}

The remainder of this paper is organized as follows. In Section~\ref{sec
formulation multishot networks}, we formulate zero-error communication
and perfect secrecy in adversarial multishot network coding. This is a
natural extension of the formulation in \cite{silva-universal} to the
multishot case. In Section~\ref{sec measures raliable and secure}, we
give sufficient and necessary conditions for coding schemes to correct a
given number of errors and erasures, and for perfect secrecy under a
given number of wire-tapped links. This again is a natural extension of
\cite{on-metrics, silva-universal}, and the sufficient condition for
error correction was already given in \cite{multishot}. In
Section~\ref{sec optimal constructions}, we provide the above-mentioned
constructions based on the linearized Reed-Solomon codes introduced in
\cite{linearizedRS}. We then give Singleton-type bounds, and prove that
our schemes attain them for coherent communication and are close to them
in the non-coherent case. Finally, we give in Section~\ref{sec
Berlekamp-Welch} a Welch-Berlekamp sum-rank decoding algorithm for
linearized Reed-Solomon codes.

\subsection*{Notation and Preliminaries}

Throughout this paper, we will fix a prime power $ q $ and positive
integers $ m $, $ \ell $, $ n = n_1 + n_2 + \cdots + n_\ell $, and $ N =
N_1 + N_2 + \cdots + N_\ell $. The number $ m $ corresponds to the
packet length in each shot, $ \ell $ is the number of shots, and $ n_i $
and $ N_i $ are the number of outgoing links at the source and incoming
links to the sink, respectively, in the $ i $th shot.

We will denote by $ \mathbb{F}_q $ the finite field with $ q $ elements. For a field $ \mathbb{F} $, we will denote by $ \mathbb{F}^{m \times n} $ the set of $ m
\times n $ matrices with entries in $ \mathbb{F} $, and we denote $ \mathbb{F}^n = \mathbb{F}^{1
\times n} $. The field over which we consider linearity, ranks, and
dimensions will be clear from the context. We will implicitly consider
``erased'' matrices, which may be denoted by $ * $. We will define $
\mathbb{F}^{m \times 0} = \{ * \} $ and $ \mathbb{F}^{0 \times n}
= \{ * \} $. Operations with matrices are assumed to be extended to $ * $
in the obvious way. For instance, $ \Rk(*) = 0 $, $ A* = *
$ for $ A \in \mathbb{F}^{n \times m} $,
etc. 

For matrices $ A_i \in \mathbb{F}^{N_i \times n_i} $, for $ i = 1,2, \ldots, \ell $, we define the \textit{block-diagonal matrix}
$$ \diag(A_1, A_2,\ldots, A_\ell) = \left( \begin{array}{cccc}
A_1 & 0 & \ldots & 0 \\
0 & A_2 & \ldots & 0 \\
\vdots & \vdots & \ddots & \vdots \\
0 & 0 & \ldots & A_\ell \\
\end{array} \right) \in \mathbb{F}^{N \times n}. $$
To define $ \diag(A_1, A_2,\ldots, A_\ell) $, for each $ i $ such that $ A_i = * \in \mathbb{F}^{0 \times n_i} $, we only add $ n_i $ zero columns between the $ (i-1) $th and the $ (i+1) $th blocks. For example, if $ A \in \mathbb{F}^{N_1 \times n_1} $, $ * \in \mathbb{F}^{0 \times n_2} $ and $ B \in \mathbb{F}^{N_3 \times n_3} $, we define $ {\rm diag}(A,*,B) \in \mathbb{F}^{(N_1 + N_3) \times (n_1 + n_2 + n_3)} $ as putting $ n_2 $ zero columns between the first $ n_1 $ and the last $ n_3 $ columns in $ {\rm diag}(A,B) \in \mathbb{F}^{(N_1 + N_3) \times (n_1 + n_3)} $.

Fix an ordered basis $ \mathcal{A} = \{ \alpha_1, \alpha_2, \ldots,
\alpha_m \} $ of $ \mathbb{F}_{q^m} $ over $ \mathbb{F}_q $. For any
non-negative integer $ s $, we denote by $ M_\mathcal{A} :
\mathbb{F}_{q^m}^s \longrightarrow \mathbb{F}_q^{m \times s} $ the
corresponding matrix representation map, given by 
\begin{equation}
M_\mathcal{A} \left(  \mathbf{c}  \right) = \left( \begin{array}{cccc}
c_{11} & c_{12} & \ldots & c_{1s} \\
c_{21} & c_{22} & \ldots & c_{2s} \\
\vdots & \vdots & \ddots & \vdots \\
c_{m1} & c_{m2} & \ldots & c_{ms} \\
\end{array} \right),
\label{eq def matrix representation map}
\end{equation}
for $ \mathbf{c} = (c_1, c_2, \ldots, c_s) \in \mathbb{F}_{q^m}^s $, where $ c_{i,1}, c_{i,2}, \ldots, c_{i,m} \in \mathbb{F}_q $ are the unique scalars such that $ c_i = \sum_{j=1}^m \alpha_j c_{i,j} \in \mathbb{F}_{q^m} $, for $ i = 1,2, \ldots, s $. 

Given $ X \in \mathbb{F}_q^{m \times n} $, we denote by $ \Row(X)
\subseteq \mathbb{F}_q^n $ and $ \Col(X) \subseteq \mathbb{F}_q^m $
the vector spaces generated by the rows and the columns of $ X $,
respectively. For $ \mathbf{c} \in \mathbb{F}_{q^m}^n $, we denote $
\Row(\mathbf{c}) = \Row( M_\mathcal{A} (\mathbf{c}) )
\subseteq \mathbb{F}_q^n $ and $ \Col(\mathbf{c}) = \Col(
M_\mathcal{A} (\mathbf{c}) ) \subseteq \mathbb{F}_q^m $. The latter
depends on $ \mathcal{A} $, but we omit this for simplicity.

Throughout the paper, we will use the short notation $ (x_i)_{i=1}^\ell = (x_1,x_2, \ldots, x_\ell) $ for tuples, where $ x_1, x_2, \ldots, x_\ell $ could be any type of element (integers, elements of finite fields, sets, subspaces...).

\section{Reliability and Security in Adversarial Multishot Network Coding} \label{sec formulation multishot networks}

In this section, we will formulate reliability and security in multishot
network coding under a worst-case adversarial model. Our analysis
naturally extends that in \cite{silva-universal} to the multishot case
(see also \cite{errors-network, on-metrics, error-control}).

We will consider $ \ell $ \textit{shots} of a network where 
linear network coding over the finite field $ \mathbb{F}_q $ is
implemented \cite{ahlswede, linearnetwork, Koetter2003}, where we have no knowledge of or control over the network
topology or the encoding coefficients (thus random linear network coding \cite{random} may be implemented). The 
network topology may change with time, the encoding 
coefficients at each node are independent for different shots, and 
we assume no delays. The \textit{packets} sent in each shot through the links
of the network and linearly combined at each node are vectors in $
\mathbb{F}_q^m $. We assume that an adversary is able to inject error
packets in up to $ t $ links, modify transfer matrices to erase up to $
\rho $ encoded packets, and wire-tap up to $ \mu $ links, all
distributed over the $ \ell $ shots of the network as the adversary
wishes (see Fig.~\ref{fig multishot network}). Other than these
restrictions, the adversary is assumed to be omniscient (he or she knows the coding scheme used at the source, the network topology and all encoding coefficients in the network) and have
unlimited computational power. 

Fix $ i = 1,2, \ldots, \ell $. Assume that the $ i $th shot has as input a
matrix $ X_i \in \mathbb{F}_q^{m \times n_i} $. Define $ t_i \geq 0 $ and $ \rho_i \geq 0 $ as the number of link errors and erasures, respectively, that the adversary introduces, and $ \mu_i \geq 0 $ as the number of links that the adversary wire-taps, all in the $ i $th shot. These numbers are thus arbitrary with sums $ t $, $ \rho $ and $ \mu $, respectively, and only these sums are known or estimated by the source or sinks. Following the model in \cite[Sec.~III]{silva-universal}, the output to the receiver is a matrix
\[
Y_i = X_iA_i^T + E_i \in \mathbb{F}_q^{m \times N_i},
\]
for a transfer matrix $ A_i \in \mathbb{F}_q^{N_i \times n_i} $ with $ \Rk(A_i) \geq n_i - \rho_i $, and for an error matrix $ E_i \in \mathbb{F}_q^{m \times N_i} $ with $ \Rk(E_i) \leq t_i $, where equalities can be achieved. The transfer matrix $ A_i $ may be known by the receiver (\textit{coherent communication}) or not (\textit{non-coherent communication}). The non-coherent case is more realistic if random linear network coding is applied. However, we include both since, as we will see, optimal solutions for the coherent case, which is simpler to solve, will give near optimal solutions for the non-coherent case. Following the same model, the adversary obtains the matrix
\[
W_i = X_i B_i^T \in \mathbb{F}_q^{m \times \mu_i},
\]
for a wire-tap transfer matrix $ B_i \in \mathbb{F}_q^{\mu_i \times n_i} $. 

Define now the set of input codewords, the set of output words to the receiver, and the set of output messages to the wire-tapper, respectively, as follows:
\begin{equation*}
\begin{split}
\mathcal{X} & = \mathbb{F}_q^{m \times n_1} \times \mathbb{F}_q^{m \times n_2} \times \cdots \times \mathbb{F}_q^{m \times n_\ell}, \\
\mathcal{Y} & = \mathbb{F}_q^{m \times N_1} \times \mathbb{F}_q^{m \times N_2} \times \cdots \times \mathbb{F}_q^{m \times N_\ell}, \\
\mathcal{W} & = \bigcup_{\substack{0 \leq \mu_i \leq n_i \\ 1 \leq i \leq \ell}}  \mathbb{F}_q^{m \times \mu_1} \times \mathbb{F}_q^{m \times \mu_2} \times \cdots \times \mathbb{F}_q^{m \times \mu_\ell} .
\end{split}
\end{equation*}
Assume that the sent codeword is $ X \in \mathcal{X} $.
With the previous restrictions on the adversary,
in the case of coherent communication with transfer matrices $ A_i \in \mathbb{F}_q^{N_i \times n_i} $, for $ i = 1,2, \ldots, \ell $, such that
\[
\sum_{i=1}^\ell \Rk(A_i) \geq n - \rho,
\]
the output of the $ \ell $-shot network is restricted to the subset
\begin{equation*}
\begin{split}
\mathcal{Y}_X^{A_1, A_2, \ldots, A_\ell}(t) = \{ & Y \in \mathcal{Y} \mid Y_i = X_i A_i^T + E_i, \\
& E_i \in \mathbb{F}_q^{m \times N_i}, \sum_{i=1}^\ell \Rk(E_i) \leq t \}.\\
\end{split}
\end{equation*}
In the case of non-coherent communication,
the output is instead restricted to the subset 
\begin{equation*}
\begin{split}
\mathcal{Y}_X(t, \rho) = \{ & Y \in \mathcal{Y} \mid Y_i = X_i A_i^T + E_i, \\
& A_i \in \mathbb{F}_q^{N_i \times n_i}, \sum_{i=1}^\ell \Rk(A_i) \geq n - \rho, \\
& E_i \in \mathbb{F}_q^{m \times N_i}, \sum_{i=1}^\ell \Rk(E_i)
\leq t \}. \\
\end{split}
\end{equation*}
Finally, the output to the adversary is restricted to the subset
\begin{equation*}
\begin{split}
\mathcal{W}_X(\mu) = \{ & W \in \mathcal{W} \mid W_i = X_i B_i^T, \\
& B_i \in \mathbb{F}_q^{\mu_i \times n_i}, \sum_{i=1}^\ell \mu_i = \mu \}.
\end{split}
\end{equation*}

\begin{remark}
Observe that by the matrix representation map (\ref{eq def matrix
representation map}), we may consider $ \mathcal{X} = \mathbb{F}_{q^m}^n
$ and $ \mathcal{Y} = \mathbb{F}_{q^m}^N $, which are vector spaces over
$ \mathbb{F}_{q^m} $. We may also consider $ \mathcal{W} =
\mathbb{F}_{q^m}^\mu $, but the union symbol above expresses the fact
that we do not know how many links the wire-tapper observes in each
shot.
\end{remark}

In this work, we will consider coding schemes as follows. 

\begin{definition}[\textbf{Coding schemes}] \label{def coding schemes}
Let $ \mathcal{S} $ be the set of secret messages and let $ S
\in \mathcal{S} $ be a random variable in $ \mathcal{S} $. A coding
scheme is a randomized function $ F : \mathcal{S} \longrightarrow
\mathcal{X} $. For unique decoding, we assume that $ \mathcal{X}_S \cap
\mathcal{X}_T = \varnothing $ if $ S,T \in \mathcal{S} $ and $ S \neq T
$, where 
\[
\mathcal{X}_S = \{ X \in \mathcal{X} \mid P(X \mid S) > 0 \},
\]
for $ S \in \mathcal{S} $. We define the support scheme of $ F $ as
the collection of disjoint sets $ {\rm Supp}(F) = \{ \mathcal{X}_S
\}_{S \in \mathcal{S}} $.
\end{definition}

As in \cite{silva-universal}, we will require zero-error communication
and perfect secrecy. We can adapt the definition of universal reliable
and secure coding schemes from \cite{silva-universal} to multishot
network coding as follows. 

\begin{definition} \label{def error correction and security}
With notation as in the previous definition and for integers $ t \geq 0
$, $ \rho \geq 0 $ and $ \mu \geq 0 $, we say that a coding scheme $ F :
\mathcal{S} \longrightarrow \mathcal{X} $ is 
\begin{enumerate}
\item $ t $-error and $ \rho $-erasure-correcting for coherent
communication if for all transfer matrices $ A_i \in \mathbb{F}_q^{N_i
\times n_i} $, for $ i = 1,2, \ldots, \ell $, with $ \sum_{i=1}^\ell
\Rk(A_i) \geq n - \rho $, there exists a decoding function $ D :
\mathcal{Y} \longrightarrow \mathcal{S} $, which may depend on $ A_1,
A_2, \ldots, A_\ell $, such that
\[
D(Y) = S,
\]
for all
$ Y \in \bigcup_{X \in \mathcal{X}_S} \mathcal{Y}_X^{A_1, A_2, \ldots,
A_\ell}(t) $ and all $ S \in \mathcal{S} $. 
\item $ t $-error and $ \rho $-erasure-correcting for non-coherent
communication if there exists a decoding function $ D : \mathcal{Y}
\longrightarrow \mathcal{S} $ such that
\[
D(Y) = S,
\]
for all $ Y \in 
\bigcup_{X \in \mathcal{X}_S} \mathcal{Y}_X(t,\rho) $ and all $ S \in
\mathcal{S} $. 
\item Secure under $ \mu $ observations if it holds that
\[
H(S \mid W) = H(S),
\]
or equivalently $ I(S;W) = 0 $, for all $ W \in \bigcup_{X \in
\mathcal{X}_S} \mathcal{W}_X(\mu) $ and all $ S \in \mathcal{S} $.
\end{enumerate}
\end{definition}

We conclude by recalling how to construct coding schemes using nested
linear code pairs. The idea goes back to \cite{ozarow} for the wire-tap
channel of type II.

\begin{definition}[\textbf{Nested coset coding schemes}] \label{definition NLCP}
Given nested linear codes $ \mathcal{C}_2 \subsetneqq \mathcal{C}_1
\subseteq \mathbb{F}_{q^m}^n $ and $ \mathcal{S} = \mathbb{F}_{q^m}^k $,
where $ k = \dim(\mathcal{C}_1) - \dim(\mathcal{C}_2) $, we define its
corresponding coding scheme $ F : \mathcal{S} \longrightarrow
\mathcal{X} $ as follows. Choose a vector space $ \mathcal{V} $ such
that $ \mathcal{C}_1 = \mathcal{C}_2 \oplus \mathcal{V} $, where $
\oplus $ denotes the direct sum of vector spaces, and a vector space
isomorphism $ \psi : \mathbb{F}_{q^m}^k \longrightarrow \mathcal{V} $.
Then we define $ F $ as any randomized function such that $
\mathcal{X}_S = \psi(S) + \mathcal{C}_2 $, for $ S \in
\mathbb{F}_{q^m}^k $. 
\end{definition}

Observe that $ \mathcal{C}_2 = \{ \mathbf{0} \} $ corresponds to using a
single linear code and a linear deterministic coding scheme. This is
suitable for reliability but not for security. The other extreme case, $
\mathcal{C}_1 = \mathbb{F}_{q^m}^n $, is suitable for security but not
for reliability.

\section{Measures of Reliability and Security}
\label{sec measures raliable and secure}

In this section, we will extend the studies in \cite{on-metrics,
silva-universal} regarding what parameter of a coding scheme gives a
necessary and sufficient condition for it to be $ t $-error and $ \rho
$-erasure-correcting, or secure under $ \mu $ observations. In
Subsection~\ref{subsec measuring coherent}, we study error and erasure
correction for coherent communication. In Subsection~\ref{subsec
measuring non-coherent}, we study error and erasure correction for
non-coherent communication. Finally, in Subsection~\ref{subsec measuring
info leakage}, we study security resistance against a wire-tapper. Since
the proofs are natural extensions of those in \cite{on-metrics,
silva-universal}, we have moved them to the appendix.

\subsection{The Sum-rank Metric and Coherent Communication} \label{subsec measuring coherent}

We start by defining the sum-rank metric in $ \mathbb{F}_{q^m}^n $,
which was introduced in \cite{multishot}. It  was  implicitly  considered  earlier
in the space-time coding literature (see \cite[Sec. III]{space-time-kumar}).

\begin{definition} [\textbf{Sum-rank metric \cite{multishot}}]
Let $ \mathbf{c} = (\mathbf{c}^{(1)}, $ $ \mathbf{c}^{(2)}, $ $ \ldots,
$ $ \mathbf{c}^{(\ell)}) \in \mathbb{F}_{q^m}^n $, where $
\mathbf{c}^{(i)} \in \mathbb{F}_{q^m}^{n_i} $, for $ i = 1,2, \ldots,
\ell $. We define the sum-rank weight of $ \mathbf{c} $ as
$$ \wt_{SR}(\mathbf{c}) = \sum_{i=1}^\ell {\rm
Rk}(M_{\mathcal{A}}(\mathbf{c}^{(i)})), $$
where $ M_\mathcal{A} $ is as in (\ref{eq def matrix representation map}), taking $ s = n_i $ for $ \mathbf{c}^{(i)} $. Finally, we define the
sum-rank metric $ \dd_{SR} : (\mathbb{F}_{q^m}^n)^2 \longrightarrow
\mathbb{N} $ as $$ \dd_{SR}(\mathbf{c}, \mathbf{d}) = {\rm
wt}_{SR}(\mathbf{c} - \mathbf{d}), $$ for all $ \mathbf{c}, \mathbf{d}
\in \mathbb{F}_{q^m}^n $.
\end{definition}

Observe that indeed the sum-rank metric is a metric, and in particular, it satisfies the triangle inequality. Observe also that sum-rank weights and metrics depend on the decomposition $ n = n_1 + n_2 + \cdots + n_\ell $ and the subfield $ \mathbb{F}_q \subseteq \mathbb{F}_{q^m} $. However, we do not write this dependency for brevity. 

Next we define minimum sum-rank distances of supports of coding schemes.

\begin{definition} [\textbf{Minimum sum-rank distance}]
Given a coding scheme $ F : \mathcal{S} \longrightarrow \mathcal{X} $,
we define its minimum sum-rank distance as that of its support scheme (see Definition \ref{def coding schemes}):
$$ \dd_{SR}({\rm Supp}(F)) = \min \{ \dd_{SR}(\mathbf{c},
\mathbf{d}) \mid \mathbf{c} \in \mathcal{X}_S, \mathbf{d} \in
\mathcal{X}_T, S \neq T \}. $$ When $ F $ is constructed from nested
linear codes $ \mathcal{C}_2 \subsetneqq \mathcal{C}_1 \subseteq
\mathbb{F}_{q^m}^n $, its minimum sum-rank distance is the relative
minimum sum-rank distance of the codes: $$ \dd_{SR}(\mathcal{C}_1,
\mathcal{C}_2) = \min \{ \wt_{SR}(\mathbf{c}) \mid \mathbf{c} \in
\mathcal{C}_1 \setminus \mathcal{C}_2 \}. $$ The minimum sum-rank
distance of a single linear code $ \mathcal{C} \subseteq
\mathbb{F}_{q^m}^n $ is $ \dd_{SR} (\mathcal{C}) = {\rm
d}_{SR}(\mathcal{C}, \{ \mathbf{0} \}) $.
\end{definition}

We now give a sufficient and necessary condition for coding schemes to
be $ t $-error and $ \rho $-erasure-correcting for coherent
communication, for any non-negative integers $ t $ and $ \rho $. The
sufficient part has already been proven in \cite[Th.~1]{multishot} for a
deterministic code. We will make use of both implications in the proof
of Theorem~\ref{th upper bound general coding schemes}.

\begin{theorem} \label{theorem correction capability}
Given integers $ t \geq 0 $ and $ \rho \geq 0 $, a coding scheme $ F :
\mathcal{S} \longrightarrow \mathcal{X} $ is $ t $-error and $ \rho
$-erasure-correcting for coherent communication if, and only if, $$ 2t +
\rho < \dd_{SR}({\rm Supp}(F)). $$
\end{theorem} 
\begin{proof}
See Appendix~\ref{app proofs for section on measures}.
\end{proof}

\subsection{The Sum-injection and Sum-subspace Distances and Non-coherent Communication} \label{subsec measuring non-coherent}

Let $ \mathbf{n} = (n_1, n_2, \ldots, n_\ell) $. We will consider the
Cartesian product lattice $$ \mathcal{P}(\mathbb{F}_q^\mathbf{n}) =
\mathcal{P}(\mathbb{F}_q^{n_1}) \times \mathcal{P}(\mathbb{F}_q^{n_2})
\times \cdots \times \mathcal{P}(\mathbb{F}_q^{n_\ell}), $$ where $
\mathcal{P}(\mathbb{F}_q^{n_i}) $ is the lattice of vector subspaces of
$ \mathbb{F}_q^{n_i} $ (that is, the collection of all vector subspaces of $ \mathbb{F}_q^{n_i} $ considered with sums and intersections of vector subspaces as operations), for $ i = 1,2, \ldots, \ell $. We also denote $
\mathcal{P}(\mathbb{F}_q^m)^\ell $ for $  \mathbf{n} =  (m,m, \ldots,m) $, of length $
\ell $. These Cartesian products were first considered for multishot
network coding in \cite[Subsec.~II-B]{multishot2009}.

As in the single shot case \cite{errors-network}, the actual
information preserved in a non-coherent multishot linear coded network
is the list of column spaces of the transmitted matrices. We may justify
this by extending the argument in \cite[Subsec.~V-A]{on-metrics}. We
start by connecting codewords in $ \mathbb{F}_{q^m}^n $ to elements in
the lattice $ \mathcal{P}(\mathbb{F}_q^m)^\ell $.

\begin{definition} \label{def sum column space}
Given $ \mathbf{c} = (\mathbf{c}^{(1)}, \mathbf{c}^{(2)}, \ldots,
\mathbf{c}^{(\ell)}) \in \mathbb{F}_{q^m}^n $, where $ \mathbf{c}^{(i)}
\in \mathbb{F}_{q^m}^{n_i} $, for $ i = 1,2, \ldots, \ell $, we define
$$ \Col_{\Sigma} (\mathbf{c} ) = (\Col ( \mathbf{c}^{(i)}))_{i=1}^\ell \in
\mathcal{P}(\mathbb{F}_q^m)^\ell. $$ For a coding scheme $ F :
\mathcal{S} \longrightarrow \mathcal{X} $, we define its sum-subspace
support scheme as $ \Col_{\Sigma} ({\rm Supp}(F)) = \{ {\rm Col}_{\Sigma}
(\mathcal{X}_S) \}_{S \in \mathcal{S}} $, where $$ \Col_{\Sigma}
(\mathcal{X}_S) = \{ \Col_{\Sigma} (\mathbf{c}) \mid \mathbf{c} \in
\mathcal{X}_S \}. $$
\end{definition}

The error and erasure-correction capability of coding schemes will be
measured by an extension of the injection distance \cite[Def.~2]{on-metrics} to the multishot scenario. 

\begin{definition} [\textbf{Sum-injection distance}] \label{def sum-injection distance}
Given $ \boldsymbol{\mathcal{U}} = (\mathcal{U}_1, $ $ \mathcal{U}_2, $ $ \ldots, \mathcal{U}_\ell) $, $ \boldsymbol{\mathcal{V}} = (\mathcal{V}_1, \mathcal{V}_2, \ldots, \mathcal{V}_\ell) \in \mathcal{P}(\mathbb{F}_q^m)^\ell $, we define their sum-injection distance as
\begin{equation*}
\begin{split}
 \dd_{SI}(\boldsymbol{\mathcal{U}}, \boldsymbol{\mathcal{V}}) = & \sum_{i=1}^\ell \dd_I(\mathcal{U}_i, \mathcal{V}_i) \\
= & \sum_{i=1}^\ell (\dim(\mathcal{U}_i + \mathcal{V}_i) -  \min \{ \dim(\mathcal{U}_i), \dim(\mathcal{V}_i) \} ) \\
= & \sum_{i=1}^\ell (\max \{ \dim(\mathcal{U}_i), \dim(\mathcal{V}_i) \} - \dim(\mathcal{U}_i \cap \mathcal{V}_i)) .
\end{split}
\end{equation*}
For a coding scheme $ F : \mathcal{S} \longrightarrow \mathcal{X} $, we define its minimum sum-injection distance as
\begin{equation*}
\begin{split}
\dd_{SI}(\Col_{\Sigma} ( {\rm Supp}(F))) = \min \{ & \dd_{SI}(\Col_{\Sigma} ( \mathbf{c} ), \Col_{\Sigma} ( \mathbf{d} )) \\
& \mid \mathbf{c} \in \mathcal{X}_S, \mathbf{d} \in \mathcal{X}_T, S \neq T \}.
\end{split}
\end{equation*}
\end{definition} 

The following result extends \cite[Th.~20]{on-metrics} to the multishot scenario.

\begin{theorem} \label{th correction capability injection}
Given integers $ t \geq 0 $ and $ \rho \geq 0 $, a coding scheme $ F : \mathcal{S} \longrightarrow \mathcal{X} $ is $ t $-error and $ \rho $-erasure-correcting for non-coherent communication if, and only if,
$$ 2t + \rho < \dd_{SI}(\Col_{\Sigma} ( {\rm Supp}(F) )). $$
In particular, it must hold that $ \Col_{\Sigma} ( \mathbf{c}) \neq \Col_{\Sigma} ( \mathbf{d} ) $ if $ \mathbf{c} \in \mathcal{X}_S $, $ \mathbf{d} \in \mathcal{X}_T $ and $ S \neq T $.
\end{theorem}
\begin{proof}
See Appendix~\ref{app proofs for section on measures}.
\end{proof}

As in the single shot case \cite{errors-network}, the number of packets
injected in the $ i $th shot by a codeword $ \mathbf{c} $, with $ {\rm
Col}_{\Sigma} (\mathbf{c}) = \boldsymbol{\mathcal{U}} \in
\mathcal{P}(\mathbb{F}_q^m)^\ell $, coincides with $ \dim(\mathcal{U}_i)
$. If all codewords inject the same number of packets in a given shot,
we may say that the coding scheme is \textit{sum-constant-dimension}.
For such coding schemes, the sum-injection distance coincides with the
sum-subspace distance introduced in \cite[Eq.~(2)]{multishot2009}.

\begin{definition}[\textbf{Sum-subspace distance \cite{multishot2009}}]
Given $ \boldsymbol{\mathcal{U}} $, $ \boldsymbol{\mathcal{V}} \in \mathcal{P}(\mathbb{F}_q^m)^\ell $, we define their sum-subspace distance as
\begin{equation*}
\begin{split}
\dd_{SS}(\boldsymbol{\mathcal{U}}, \boldsymbol{\mathcal{V}}) = & \sum_{i=1}^\ell \dd_S(\mathcal{U}_i, \mathcal{V}_i) \\
= & \sum_{i=1}^\ell (\dim(\mathcal{U}_i + \mathcal{V}_i) - \dim(\mathcal{U}_i \cap \mathcal{V}_i)).
\end{split}
\end{equation*}
For a coding scheme $ F : \mathcal{S} \longrightarrow \mathcal{X} $, we define its minimum sum-subspace distance, denoted by $ \dd_{SS}(\Col_{\Sigma} ( {\rm Supp}(F) )) $, analogously to their minimum sum-injection distance (Definition~\ref{def sum-injection distance}).
\end{definition}

Observe that by \cite[Eq.~(28)]{on-metrics}, it holds that
\begin{equation}
\dd_{SI}(\boldsymbol{\mathcal{U}}, \boldsymbol{\mathcal{V}}) = \frac{1}{2} \dd_{SS}(\boldsymbol{\mathcal{U}}, \boldsymbol{\mathcal{V}}) + \frac{1}{2} \sum_{i=1}^\ell | \dim(\mathcal{U}_i) - \dim(\mathcal{V}_i) |,
\label{eq comparison d_SS and d_SI}
\end{equation}
for all $ \boldsymbol{\mathcal{U}} $, $ \boldsymbol{\mathcal{V}} \in \mathcal{P}(\mathbb{F}_q^m)^\ell $. Hence as explained earlier, for sum-constant-dimension codes we may simply consider its minimum sum-subspace distance. Then the sufficient part in Theorem~\ref{th correction capability injection} was already proven in \cite[Th.~1]{multishot}.

\subsection{Measuring Security Resistance Against a Wire-tapper} \label{subsec measuring info leakage}

In this subsection, we give a sufficient and necessary condition for coding schemes built from nested linear code pairs to be secure under a given number of observations. We start by estimating the information leakage to the wire-tapper. This is a natural extension of \cite[Lemma~6]{silva-universal}. We will however follow the steps in the proof of \cite[Prop.~16]{rgmw} (see also \cite[Lemma~7]{rgrw}). 

\begin{lemma} \label{lemma information leakage calculation}
Let $ F : \mathcal{S} \longrightarrow \mathcal{X} $ be a coding scheme built from nested linear codes $ \mathcal{C}_2 \subsetneqq \mathcal{C}_1 \subseteq \mathbb{F}_{q^m}^n $ as in Definition~\ref{definition NLCP}. Assume that $ X = F(S) $ is the uniform random variable in $\mathcal{X}_S $ given $ S \in \mathcal{S} $. Let $ B_i \in \mathbb{F}_q^{\mu_i \times n_i} $ and $ \mathcal{L}_i = \Row( B_i ) \subseteq \mathbb{F}_q^{n_i} $, for $ i = 1,2, \ldots, \ell $, define $ B = \diag(B_1, B_2, \ldots, B_\ell) \in \mathbb{F}_q^{\mu \times n} $, $ \boldsymbol{\mathcal{L}} = (\mathcal{L}_1, \mathcal{L}_2, \ldots, \mathcal{L}_\ell) \in \mathcal{P}(\mathbb{F}_q^\mathbf{n}) $ and
$$ \mathcal{V}_{\boldsymbol{\mathcal{L}}} = \{ \mathbf{c} \in \mathbb{F}_{q^m}^n \mid \Row( \mathbf{c}^{(i)}) \subseteq \mathcal{L}_i, 1 \leq i \leq \ell \}. $$ 
Taking logarithms with base $ q^m $ in entropy and mutual information, it holds that
\begin{equation}
I(S; X B^T) \leq \dim(\mathcal{C}_2^\perp \cap \mathcal{V}_{\boldsymbol{\mathcal{L}}}) - \dim(\mathcal{C}_1^\perp \cap \mathcal{V}_{\boldsymbol{\mathcal{L}}}),
\label{eq information leakage equation}
\end{equation}
and equality holds if $ S $ is the uniform random variable in $ \mathcal{S} $.
\end{lemma}
\begin{proof} 
See Appendix~\ref{app proofs for section on measures}.
\end{proof}

Thus we may conclude the following.

\begin{theorem} \label{th sum-rank measures security}
Let $ F : \mathcal{S} \longrightarrow \mathcal{X} $ be a coding scheme built from nested linear codes $ \mathcal{C}_2 \subsetneqq \mathcal{C}_1 \subseteq \mathbb{F}_{q^m}^n $ as in Definition~\ref{definition NLCP}. Assume that $ F(S) $ is the uniform random variable in $ \mathcal{X}_S $ given $ S \in \mathcal{S} $. Given an integer $ \mu \geq 0 $, the coset coding scheme is secure under $ \mu $ observations if
$$ \mu < \dd_{SR}(\mathcal{C}_2^\perp, \mathcal{C}_1^\perp). $$
The reversed implication also holds if $ S $ is the uniform random variable in $ \mathcal{S} $.
\end{theorem}
\begin{proof} 
This follows by combining the previous lemma with the straightforward fact that $ \mu < \dd_{SR}(\mathcal{C}_2^\perp, \mathcal{C}_1^\perp) $ if, and only if, $ \mathcal{C}_2^\perp \cap \mathcal{V}_{\boldsymbol{\mathcal{L}}} = \mathcal{C}_1^\perp \cap \mathcal{V}_{\boldsymbol{\mathcal{L}}} $ for all $ \boldsymbol{\mathcal{L}} \in \mathcal{P}(\mathbb{F}_q^\mathbf{n}) $ such that $ \sum_{i=1}^\ell \dim(\mathcal{L}_i) \leq \mu $.
\end{proof}

\section{Coding Schemes Based on Linearized Reed-Solomon Codes} \label{sec optimal constructions}

In this section, we will introduce a family of coding schemes whose
secret message entropy is the maximum possible for coherent communication,
for any  fixed bound on the number of errors, erasures and wire-tapped links, with the
restrictions $ 1 \leq \ell \leq q-1 $ and $ 1 \leq n_i \leq m $, for $ i
= 1,2, \ldots, \ell $. In particular, $ n \leq (q-1)m $. Assuming the
secret message is a uniform variable, these coding schemes have maximum
secret message size for the largest possible range of packet lengths.

By a process analogous to lifting \cite[Subsec.~IV-A]{error-control}, we
will also provide coding schemes with nearly optimal information rate
for non-coherent communication, for the same parameters, among
sum-constant-dimension codes. 

The building blocks of these coding schemes are linearized Reed-Solomon
codes, introduced in \cite[Def.~31]{linearizedRS}. We review these codes
in Subsection~\ref{subsec lin RS codes and duals}, and we compute their
duals. We then give upper bounds on the message entropy or size and show
that the above-mentioned coding schemes attain them in the coherent case
(Subsection~\ref{subsec optimal for coherent}), and are close in the
non-coherent case (Subsection~\ref{subsec near optimal non-coherent}).

\subsection{Linearized Reed-Solomon Codes and their Duals}
\label{subsec lin RS codes and duals}

Throughout this subsection, we assume that $ 1 \leq \ell \leq q-1 $ and
$ 1 \leq n_i \leq m $, for $ i = 1,2, \ldots , \ell $. Therefore $ n
\leq (q-1)m $. 

Let $ \sigma : \mathbb{F}_{q^m} \longrightarrow \mathbb{F}_{q^m} $ be
the field automorphism given by $ \sigma(a) = a^{q^r} $, for all $ a \in
\mathbb{F}_{q^m} $, where $ 1 \leq r \leq m $ and $ \gcd(r,m) = 1 $
(i.e., $ \{ a \in \mathbb{F}_{q^m} \mid \sigma(a) = a \} = \mathbb{F}_q $). We may set for
simplicity $ r = 1 $, but we need to consider the general case later to
include dual codes. We need to define linear operators as in
\cite[Def.~20]{linearizedRS}. Over finite fields (our case), polynomials
in these operators can be regarded as linearized polynomials
\cite[Ch.~3]{lidl}.

\begin{definition}[\textbf{Linear operators \cite{linearizedRS}}] \label{def linearized operators}
Fix $ a \in \mathbb{F}_{q^m} $, and define its $ i $th norm as $ N_i(a)
= \sigma^{i-1}(a) \cdots \sigma(a)a $. Now define the $ \mathbb{F}_q
$-linear operator $ \mathcal{D}_a^i : \mathbb{F}_{q^m} \longrightarrow
\mathbb{F}_{q^m} $ by
\begin{equation}
\mathcal{D}_a^i(b) = \sigma^i(b) N_i(a) ,
\label{eq definition linear operator}
\end{equation}
for all $ b \in \mathbb{F}_{q^m} $, and all $ i \in \mathbb{N} $. Define also $ \mathcal{D}_a = \mathcal{D}_a^1 $ and observe that $ \mathcal{D}_a^{i+1} = \mathcal{D}_a \circ \mathcal{D}_a^i $, for $ i \in \mathbb{N} $. We will write $ N^\sigma_i $ and $ \mathcal{D}_a^\sigma $ when it is not clear which automorphism $ \sigma $ we are using.
\end{definition}

We also need the concept of conjugacy in $ \mathbb{F}_{q^m} $, which was
given in \cite{lam} (see also \cite[Eq.~(2.5)]{lam-leroy}).

\begin{definition} [\textbf{Conjugacy \cite{lam}}] \label{def conjugacy}
Given $ a,b \in \mathbb{F}_{q^m} $, we say that they are conjugates if
there exists $ c \in \mathbb{F}_{q^m}^* $ such that $$ b =
\sigma(c)c^{-1} a. $$
\end{definition}

This defines an equivalence relation on $ \mathbb{F}_{q^m} $, and thus a
partition of $ \mathbb{F}_{q^m} $ into conjugacy classes. Take now a
primitive element $ \gamma $ of $ \mathbb{F}_{q^m} $, meaning that $ \mathbb{F}_{q^m}^* = \{ \gamma^0, \gamma^1, \gamma^2 \ldots, \gamma^{q^m-2} \} $ (see \cite[page 97]{macwilliams}), and observe that
$$ \gamma^j \neq \sigma(c)c^{-1} \gamma^i, $$ for all $ c \in
\mathbb{F}_{q^m}^* $ and all $ 0 \leq i < j \leq q-2 $. Hence $
\gamma^0, $ $ \gamma^1, $ $ \gamma^2, $ $ \ldots , $ $ \gamma^{q-2} $
constitute representatives of disjoint conjugacy classes. Moreover, one
can easily see that they represent all conjugacy classes except the
trivial one $ \{ 0 \} $. Thus the following definition is a particular
case of \cite[Def.~31]{linearizedRS}.

\begin{definition} [\textbf{Linearized Reed-Solomon codes \cite{linearizedRS}}] \label{def lin RS codes}
Fix linearly independent sets $ \mathcal{B}^{(i)} = \{ \beta_1^{(i)},
\beta_2^{(i)}, \ldots, \beta_{n_i}^{(i)} \} \subseteq \mathbb{F}_{q^m} $
over $ \mathbb{F}_q $, for $ i = 1,2, \ldots, \ell $, and denote $
\boldsymbol{\mathcal{B}} = (\mathcal{B}^{(1)}, $ $ \mathcal{B}^{(2)}, $
$ \ldots, $ $ \mathcal{B}^{(\ell)}) $. Here, we emphasize that the elements in $ \mathcal{B}^{(i)} $ are linearly independent, but there may be linear dependencies between elements in $ \mathcal{B}^{(i)} $ and elements in $ \bigcup_{j \neq i} \mathcal{B}^{(j)} $. In particular, it is possible to take $ \mathcal{B}^{(1)} = \mathcal{B}^{(2)} = \ldots = \mathcal{B}^{(\ell)} $. Let $ \gamma $ be a  primitive
element of $ \mathbb{F}_{q^m} $. For $ k =
0,1,2, \ldots, n $, we define the linearized Reed-Solomon code of
dimension $ k $ as the linear code $
\mathcal{C}^\sigma_{L,k}(\boldsymbol{\mathcal{B}},\gamma) \subseteq
\mathbb{F}_{q^m}^n $ with generator matrix given by
\begin{equation}
D = (D_1| D_2 | \ldots | D_\ell) \in \mathbb{F}_{q^m}^{k \times n},
\label{eq gen matrix of lin RS code}
\end{equation}
where
$$ D_i = \left( \begin{array}{cccc}
\beta_1^{(i)} & \beta_2^{(i)} & \ldots & \beta_{n_i}^{(i)} \\
\mathcal{D}_{\gamma^{i-1}} \left( \beta_1^{(i)} \right) & \mathcal{D}_{\gamma^{i-1}} \left( \beta_2^{(i)} \right) & \ldots & \mathcal{D}_{\gamma^{i-1}} \left( \beta_{n_i}^{(i)} \right) \\
\mathcal{D}_{\gamma^{i-1}}^2 \left( \beta_1^{(i)} \right) & \mathcal{D}_{\gamma^{i-1}}^2 \left( \beta_2^{(i)} \right) & \ldots & \mathcal{D}_{\gamma^{i-1}}^2 \left( \beta_{n_i}^{(i)} \right) \\
\vdots & \vdots & \ddots & \vdots \\
\mathcal{D}_{\gamma^{i-1}}^{k-1} \left( \beta_1^{(i)} \right) & \mathcal{D}_{\gamma^{i-1}}^{k-1} \left( \beta_2^{(i)} \right) & \ldots & \mathcal{D}_{\gamma^{i-1}}^{k-1} \left( \beta_{n_i}^{(i)} \right) \\
\end{array} \right) $$
for $ i = 1,2, \ldots, \ell $.
\end{definition}

By \cite[Prop.~33]{linearizedRS} (see also Subsection~\ref{subsec skew metrics}), this code is isomorphic as a vector space to a $ k $-dimensional skew Reed-Solomon code \cite[Def.~7]{skew-evaluation1}. In particular, it is also $ k $-dimensional and the generator matrix in (\ref{eq gen matrix of lin RS code}) has full rank. Moreover, linearized Reed-Solomon codes are maximum sum-rank distance codes, which was proven in \cite{linearizedRS}. The next result combines \cite[Prop.~34]{linearizedRS} and \cite[Th.~4]{linearizedRS}.

\begin{proposition} [\textbf{\cite{linearizedRS}}] \label{prop linearized RS are MSRD}
For a linear code $ \mathcal{C} \subseteq \mathbb{F}_{q^m}^n $ of dimension $ k $, it holds that
$$ \dd_{SR}(\mathcal{C}) \leq n - k + 1. $$
Moreover, equality holds for $ \mathcal{C} = \mathcal{C}^\sigma_{L,k}(\boldsymbol{\mathcal{B}},\gamma) $ as in Definition~\ref{def lin RS codes}. That is, linearized Reed-Solomon codes are maximum sum-rank distance (MSRD) codes.
\end{proposition}

Observe that Gabidulin codes \cite{gabidulin, roth} and their extension \cite{new-construction} are obtained as particular cases by choosing $ \ell = 1 $ (thus $ n = n_1 $): In that case, we have that $ N_j(\gamma^0) = N_j(1) = 1 $, hence 
$$ \mathcal{D}_{\gamma^{0}}^j(\beta_i^{(1)}) = \sigma^j(\beta_i^{(1)}), $$ 
for all $ i = 1,2,\ldots, n $ and $ j = 0,1,\ldots, k-1 $. 

One can also recover Reed-Solomon codes \cite{reed-solomon} and generalized Reed-Solomon codes \cite{delsarte} by setting $ \sigma = \Id $ or $ m=1 $ (thus $ n_1 = n_2 = \ldots = n_\ell = 1 $): In that case, we have that $ N_j(\gamma^{i-1}) = (\gamma^{i-1})^j $ and $ \sigma(\beta_1^{(i)}) = \beta_1^{(i)} $, hence
$$ \mathcal{D}_{\gamma^{i-1}}^j(\beta_1^{(i)}) = \beta_1^{(i)} (\gamma^{i-1})^j, $$ 
for $ i = 1,2,\ldots, \ell $ and $ j = 0,1,\ldots, k-1 $. This explains the discussion in Section~\ref{sec intro}.

Intuitively, the conjugacy representative $ \gamma^{i-1} $ and the
linear operator $ \mathcal{D}_{\gamma^{i-1}} $ are used
in the $ i $th shot of the network as a Gabidulin code, whereas using
different conjugacy classes allows to correct link errors globally,
seeing the code block-wise as a Reed-Solomon code (observe the
correspondence between Fig.~\ref{fig multishot network} and the generator matrix (\ref{eq gen matrix of lin RS code})). See also Table \ref{table skew ev codes} in Subsection \ref{subsec skew metrics} for a summary of maximum-distance evaluation codes of Reed-Solomon type for different metrics.  

We now prove that the duals of these codes are again linearized Reed-Solomon codes, thus also maximum sum-rank distance, which was not proven in \cite{linearizedRS}. This result extends \cite[Th.~1]{delsarte}, \cite[Th.~7]{gabidulin} and \cite[Th.~2]{new-construction}. We will need it to obtain optimal secure coding schemes in view of Theorem~\ref{th sum-rank measures security}. It also has the advantage of providing explicit parity-check matrices for linearized Reed-Solomon codes in the form of (\ref{eq gen matrix of lin RS code}).

\begin{theorem} \label{th duals of lin RS codes}
For $ k = 0,1,2,\ldots, n $, if $ \mathcal{C}^\sigma_{L,k}(\boldsymbol{\mathcal{B}},\gamma) $ is as in Definition~\ref{def lin RS codes}, then there exist linearly independent sets $ \mathcal{A}^{(i)} = \{ \alpha_1^{(i)}, \alpha_2^{(i)}, \ldots, $ $ \alpha_{n_i}^{(i)} \} \subseteq \mathbb{F}_{q^m} $ over $ \mathbb{F}_q $, for $ i = 1,2, \ldots, \ell $, such that
$$ \mathcal{C}^\sigma_{L,k}(\boldsymbol{\mathcal{B}},\gamma)^\perp = \mathcal{C}^{\sigma^{-1}}_{L,n-k}(\boldsymbol{\mathcal{A}},\sigma^{-1}(\gamma)), $$
where $ \boldsymbol{\mathcal{A}} = (\mathcal{A}^{(1)}, \mathcal{A}^{(2)}, \ldots, \mathcal{A}^{(\ell)}) $, and where $ \sigma^{-1}(\gamma) $ is a primitive element of $ \mathbb{F}_{q^m} $. In particular, $ \mathcal{C}^\sigma_{L,k}(\boldsymbol{\mathcal{B}},\gamma)^\perp $ is also an MSRD code.
\end{theorem}
\begin{proof}
First $
\mathcal{C}^\sigma_{L,n-1}(\boldsymbol{\mathcal{B}},\gamma)^\perp $ is
generated by one vector $ \boldsymbol{\alpha} =
(\boldsymbol{\alpha}^{(1)}, \boldsymbol{\alpha}^{(2)}, \ldots,
\boldsymbol{\alpha}^{(\ell)}) \in \mathbb{F}_{q^m}^n $, where $
\boldsymbol{\alpha}^{(i)} \in \mathbb{F}_{q^m}^{n_i} $, for $ i = 1,2,
\ldots, \ell $. Assume that there exists $ i $ such that $
\alpha_1^{(i)}, \alpha_2^{(i)}, \ldots, \alpha_{n_i}^{(i)} $ are
linearly dependent over $ \mathbb{F}_q $. Then there exists a non-zero $
\boldsymbol{\lambda} \in \mathbb{F}_q^n $ with sum-rank weight equal to
$ 1 $, such that $ \boldsymbol{\alpha} \boldsymbol{\lambda}^T = 0 $.
This implies that $ \boldsymbol{\lambda} \in
\mathcal{C}^\sigma_{L,n-1}(\boldsymbol{\mathcal{B}},\gamma) $, but $
\dd_{SR}(\mathcal{C}^\sigma_{L,n-1}(\boldsymbol{\mathcal{B}},\gamma)) =
2 $, which is a contradiction. Thus $ \alpha_1^{(i)}, \alpha_2^{(i)},
\ldots, \alpha_{n_i}^{(i)} $ are linearly independent over $
\mathbb{F}_q $, for $ i = 1,2, \ldots, \ell $. Now by definition, and denoting $ \mathcal{D}_{\gamma^{i-1}} = \mathcal{D}_{\gamma^{i-1}}^\sigma $, we have
that $$ \sum_{i=1}^\ell \sum_{j=1}^{n_i} \alpha_j^{(i)}
\mathcal{D}_{\gamma^{i-1}}^l(\beta_j^{(i)}) = 0, $$ for $ l = 0,1,2, \ldots, n-2 $.
Take $ t = 0,1,2,\ldots, n-k-1 $. Applying the automorphism $
\sigma^{-t} $, we see that
\begin{align*}
0 & = \sum_{i=1}^\ell \sum_{j=1}^{n_i}
\sigma^{-t}(\alpha_j^{(i)}) \sigma^{-t}(
\mathcal{D}_{\gamma^{i-1}}^l(\beta_j^{(i)}))\\
& = \sum_{i=1}^\ell \sum_{j=1}^{n_i} \sigma^{-t}(\alpha_j^{(i)})
\sigma^{-t}(N^\sigma_t(\gamma^{i-1}))
\mathcal{D}_{\gamma^{i-1}}^{l-t}(\beta_j^{(i)})\\
& = \sum_{i=1}^\ell \sum_{j=1}^{n_i} \sigma^{-t}(\alpha_j^{(i)})
N^{\sigma^{-1}}_t(\sigma^{-1}(\gamma)^{i-1})
\mathcal{D}_{\gamma^{i-1}}^{l-t}(\beta_j^{(i)}) \\
& = \sum_{i=1}^\ell
\sum_{j=1}^{n_i} (\mathcal{D}^{\sigma^{-1}}_{\sigma^{-1}(\gamma)^{i-1}})^t (\alpha_j^{(i)})
\mathcal{D}_{\gamma^{i-1}}^{l-t}(\beta_j^{(i)}),
\end{align*}
for $ i =
1,2, \ldots, \ell $. By considering $ t = 0,1,2,\ldots, n-k-1 $ and $ l
= 0,1,2,\ldots, n-2 $, we may consider the number $ l-t $ to run from $
0 $ to $ k-1 $. In conclusion, we have proven that $$
\mathcal{C}^{\sigma^{-1}}_{L,n-k}(\boldsymbol{\mathcal{A}},\sigma^{-1}(\gamma))
\subseteq
\mathcal{C}^\sigma_{L,k}(\boldsymbol{\mathcal{B}},\gamma)^\perp, $$ and
by computing dimensions, equality holds. Finally, note that $ \sigma^{-1}(\gamma) $ is a primitive element since $ \gamma $ is a primitive element and $ \sigma $ is a field automorphism.
\end{proof}

Observe that the linearly independent sets $ \boldsymbol{\mathcal{A}} $ reduce in the Gabidulin case ($ \ell = 1 $) to the single linearly independent set denoted by $ \{ \lambda_1, \lambda_2, \ldots, \lambda_n \} $ in \cite[Eq.~(24)]{gabidulin}. The sets $ \boldsymbol{\mathcal{A}} $ also reduce in the Reed-Solomon case ($ n_1 = n_2 = \ldots = n_\ell = 1 $) to the non-zero column multipliers $ a_1^\prime, a_2^\prime, \ldots, a_\ell^\prime \in \mathbb{F}_q^* $ ($ \ell $ linearly independent sets in $ \mathbb{F}_q $ over $ \mathbb{F}_q $) of the dual generalized Reed-Solomon codes in \cite[Th.~1]{delsarte}.

\subsection{Optimal Coding Schemes for Coherent Communication} \label{subsec optimal for coherent}

To claim the optimality of coding schemes built from linearized Reed-Solomon codes, we will extend the upper bound on the entropy of the secret message from \cite[Th.~12]{silva-universal} to the multishot case.

\begin{theorem} \label{th upper bound general coding schemes}
Let $ t \geq 0 $, $ \rho \geq 0 $ and $ \mu \geq 0 $ be integers such that $ 2t + \rho + \mu < n $. For given random distributions (not necessarily uniform), if the coding scheme $ F : \mathcal{S} \longrightarrow \mathcal{X} $ is $ t $-error and $ \rho $-erasure-correcting for coherent communication, and secure under $ \mu $ observations, then it holds that
$$ H(S) \leq n - 2t - \rho - \mu, $$
where entropy is taken with logarithms with base $ q^m $. 
\end{theorem}
\begin{proof}
Define $ \rho^\prime = 2t + \rho $. By Theorem~\ref{theorem correction capability}, $ F $ is $ \rho^\prime $-erasure-correcting. Decompose $ \rho^\prime = \rho_1 + \rho_2 + \cdots + \rho_\ell $ and $ \mu = \mu_1 + \mu_2 + \cdots + \mu_\ell $, with $ \rho_i, \mu_i \geq 0 $ and $ \rho_i + \mu_i \leq n_i $, for $ i = 1,2, \ldots, \ell $. Let $ A_i \in \mathbb{F}_q^{(n_i - \rho_i) \times n_i} $ and $ B_i \in \mathbb{F}_q^{\mu_i \times n_i} $ be the matrices formed by the first $ n_i - \rho_i $ and $ \mu_i $ rows of the $ n_i \times n_i $ identity matrix, respectively, for $ i = 1,2, \ldots, \ell $. Define $ A = \diag(A_1, A_2, \ldots, A_\ell) \in \mathbb{F}_q^{(n - \rho^\prime) \times n} $ and $ B = \diag(B_1, B_2, \ldots, B_\ell) \in \mathbb{F}_q^{\mu \times n} $. Observe that $ B $ is a submatrix of $ A $, and let $ D \in \mathbb{F}_q^{(n - \rho^\prime - \mu) \times n} $ be the matrix formed by those rows in $ A $ that are not in $ B $.

Now let $ X = F(S) $, as random variables. Then Items 1 and 3 in Definition~\ref{def error correction and security} allow us to write
\begin{align*}
H(S) & =  H(S\mid XB^T) \\
 & \leq  H(S, XA^T \mid XB^T) \\
 &=  H(S \mid XA^T, XB^T) + H(XA^T \mid XB^T) \\
 &=  H(XA^T \mid XB^T) \\
 &\leq  H(XD^T) \\
 &\leq  n - \rho^\prime - \mu = n - 2t - \rho - \mu,
\end{align*}
and we are done.
\end{proof}

We may now claim the optimality of coding schemes built from nested pairs of linearized Reed-Solomon codes:

\begin{theorem} \label{th optimal for coherent comm}
Fix integers $ t \geq 0$, $ \rho \geq 0 $, $ \mu \geq 0 $ with $ 2t + \rho + \mu < n $, and assume that $ 1 \leq \ell \leq q-1 $ and $ 1 \leq n_i \leq m $, for $ i = 1,2, \ldots , \ell $. Let $ k_1 = n-2t-\rho $ and $ k_2 = \mu $ and define $ \mathcal{C}_1 = $ $ \mathcal{C}^\sigma_{L,k_1}(\boldsymbol{\mathcal{B}},\gamma) $ and $ \mathcal{C}_2 = \mathcal{C}^\sigma_{L,k_2}(\boldsymbol{\mathcal{B}},\gamma) $ as in Definition~\ref{def lin RS codes}. The coset coding scheme in Definition~\ref{definition NLCP}, corresponding to $ \mathcal{C}_2 \subsetneqq \mathcal{C}_1 \subseteq \mathbb{F}_{q^m}^n $ and choosing $ S $ and $ X = F(S) $ with uniform distributions, is $ t $-error and $ \rho $-erasure-correcting for coherent communication, and is secure under $ \mu $ observations. In addition, the entropy of the secret message satisfies that
$$ H(S) = \log_{q^m} |\mathcal{S}| = k_1 - k_2 = n - 2t - \rho - \mu. $$
Hence the coding scheme is optimal according to Theorem~\ref{th upper bound general coding schemes}.
\end{theorem}
\begin{proof}
This follows by combining Theorem~\ref{theorem correction capability}, Theorem~\ref{th sum-rank measures security}, Proposition~\ref{prop linearized RS are MSRD} and Theorem~\ref{th duals of lin RS codes}.
\end{proof}

\begin{remark}
In a similar way to the proof of \cite[Th.~10]{silva-universal}, we may show that if a coding scheme has secret size achieving the upper bound in Theorem~\ref{th upper bound general coding schemes}, then it must hold that $ m \geq n_i $, for $ i = 1,2, \ldots, \ell $. Thus, the coding schemes in the previous theorem are defined over the largest possible range for the packet length $ m $, which is the same as in the $ 1 $-shot case. 

However, we do not yet know if the condition $ \ell < q $ can be relaxed. As in the MDS case ($ n_1 = n_2 = \ldots = n_\ell = 1 $), we conjecture that $ \ell $ must be $ \mathcal{O}(q) $ (for fixed $ m $ and $ n_i $). We leave this bound as open problem. See also Section~\ref{sec conclusion}.
\end{remark}

\subsection{Nearly Optimal Coding schemes for Non-coherent Communication} \label{subsec near optimal non-coherent}

In this subsection, we adapt the \textit{lifting} construction from \cite[Def.~3]{error-control} to linearized Reed-Solomon codes. In contrast with the coherent case, the information rate of such construction is no longer optimal. However, we will show that it is close to optimal in reasonable scenarios in practice, as shown in \cite[Sec.~IV]{error-control} (see also \cite[Sec.~VII]{silva-universal}). We will make a few simplifying assumptions. 

First, in practical situations it is desirable to encode separately (in different layers) for secrecy and reliability, which can be done using nested coset coding schemes as in Definition~\ref{definition NLCP} (see \cite[Subsec.~VII-B]{silva-universal}). Since linearized Reed-Solomon codes give optimal security in view of Theorems~\ref{th upper bound general coding schemes} and \ref{th optimal for coherent comm}, we will only consider deterministic coding schemes and error and erasure correction. Moreover, the lifting construction will only add packet headers that contain no information about the secret message, hence they do not affect the overall security performance \cite[Subsec.~VII-E]{silva-universal}.

Second, by the fact that only column spaces are of importance for reliability in the non-coherent case (Theorem~\ref{th correction capability injection}), we will consider sum-subspace codes \cite[Subsec.~II-B]{multishot2009}. Since a codeword can inject up to $ n_i $ packets in the $ i $th shot, we will restrict our study to sum-constant-dimension codes, which are simply subsets of the Cartesian product
$$ \mathcal{P}(\mathbb{F}_q^\mathbf{M}, \mathbf{n}) = \mathcal{P}(\mathbb{F}_q^{M_1}, n_1) \times \mathcal{P}(\mathbb{F}_q^{M_2}, n_2) \times \cdots \times \mathcal{P}(\mathbb{F}_q^{M_\ell}, n_\ell), $$
where $ \mathbf{M} = (M_1, M_2, \ldots, M_\ell) $, and $ \mathcal{P}(\mathbb{F}_q^{M_i}, n_i) $ is the family of vector subspaces of $ \mathbb{F}_q^{M_i} $ of dimension $ n_i $, for $ i = 1,2, \ldots, \ell $. In view of Subsection~\ref{subsec measuring non-coherent}, we will consider the following parameters of sum-constant-dimension codes.

\begin{definition} [\textbf{Sum-constant-dimension codes}]
A sum-constant-dimension code of type $ [\mathbf{M}, \mathbf{n}, \log_q|\mathcal{C}|, d]_q $ is any subset $ \mathcal{C} \subseteq \mathcal{P}(\mathbb{F}_q^\mathbf{M}, \mathbf{n}) $ such that $ \dd_{SS}(\mathcal{C}) = 2d $. We define its rate as
$$ R = \frac{\log_q|\mathcal{C}|}{\sum_{i=1}^\ell M_i n_i}. $$
\end{definition}

For integers $ 0 \leq N \leq M $, recall that the $ q $-ary Gaussian coefficients are given by
$$ {M \brack N}_q = | \mathcal{P}(\mathbb{F}_q^M, N) | = \prod_{j=0}^{N-1} \frac{q^M - q^j}{q^N - q^j}. $$
Therefore, we have that
\begin{equation}
| \mathcal{P}(\mathbb{F}_q^\mathbf{M}, \mathbf{n}) | = \prod_{i=1}^\ell {M_i \brack n_i}_{q} .
\label{eq size of multishot grassmannian}
\end{equation}

The following sum-subspace Singleton bound is a refinement of that in \cite[Subsec.~VI-B]{multishot2009} (in the bound in \cite{multishot2009}, a single puncturing removes a whole factor $ \mathcal{P}(\mathbb{F}_q^{M_i}, n_i) $).

\begin{theorem} \label{th sum subspace singleton bound}
Let $ \mathcal{C} $ be a sum-constant-dimension code of type $ [\mathbf{M}, \mathbf{n}, \log_q|\mathcal{C}|, d]_q $. It holds that
\begin{equation}
|\mathcal{C}| \leq \min \prod_{i=1}^\ell {M_i - \delta_i \brack M_i - n_i}_{q} ,
\label{eq sum-subspace singleton bound}
\end{equation}
where the minimum is taken over numbers $ 0 \leq \delta_i \leq n_i $, for $ i = 1,2, \ldots, \ell $, such that $ \delta_1 + \delta_2 + \cdots + \delta_\ell = d-1 $.
\end{theorem}
\begin{proof}
Given $ \boldsymbol{\mathcal{U}} = (\mathcal{U}_1, \mathcal{U}_2, \ldots, \mathcal{U}_\ell) \in \mathcal{P}(\mathbb{F}_q^\mathbf{M}, \mathbf{n}) $ and a hyperplane $ \mathcal{W} \in \mathcal{P}(\mathbb{F}_q^{M_i}, M_i - 1) $, we define the restricted list of subspaces $ \boldsymbol{\mathcal{U}}_{i,\mathcal{W}} $ as that obtained from $ \boldsymbol{\mathcal{U}} $ by substituting $ \mathcal{U}_i $ by an $ (n_i - 1) $-dimensional subspace of $ \mathcal{U}_i \cap \mathcal{W} $, then mapped by a vector space isomorphism (depending only on $ \mathcal{W} $) to a subspace of $ \mathbb{F}_q^{M_i-1} $. 

Assuming that $ d > 2 $, the restricted code $ \mathcal{C}_{i,\mathcal{W}} = \{ \boldsymbol{\mathcal{U}}_{i,\mathcal{W}} \mid \boldsymbol{\mathcal{U}} \in \mathcal{C} \} $ is a sum-constant-dimension code of type $ [\mathbf{M} - \mathbf{e}_i, \mathbf{n} - \mathbf{e}_i, |\mathcal{C}|, d^\prime]_q $, where $ d^\prime \geq d-2 > 0 $ and $ \mathbf{e}_i $ is the $ i $th vector of the canonical basis. The proof of this claim is exactly as that of \cite[Th.~8]{errors-network} and is left to the reader.

Finally, applying such restriction operations to $ \mathcal{C} $, $ \delta_i $ times in the $ i $th block, we obtain a sum-constant-dimension code $ \mathcal{C}^\prime $ of type $ [\mathbf{M} - \boldsymbol{\delta}, \mathbf{n} - \boldsymbol{\delta}, \log_q |\mathcal{C}|, d^\prime]_q $, where $ d^\prime >0 $ and $ \boldsymbol{\delta} = (\delta_1, \delta_2, \ldots, \delta_\ell) $. The result follows now from (\ref{eq size of multishot grassmannian}) and the fact that 
$$ |\mathcal{C}| = |\mathcal{C}^\prime| \leq |\mathcal{P}(\mathbb{F}_q^{\mathbf{M} - \boldsymbol{\delta}}, \mathbf{n} - \boldsymbol{\delta})|. $$
\end{proof}

We will not simplify further the minimum in (\ref{eq sum-subspace singleton bound}), since it will not be necessary for our purposes in Theorem~\ref{th near optimal coding schemes non-coherent}.

We now recall the extension of the lifting procedure from \cite[Def.~3]{error-control} to the multishot scenario, which was first considered in \cite[Subsec.~III-D]{multishot}. In the rest of the subsection, we will assume that $ M_i = n_i + m $, for $ i = 1,2, \ldots, \ell $.

\begin{definition} [\textbf{Lifting \cite{multishot}}] \label{def lifting}
We define the lifting map $ \mathcal{I}_{\Sigma} : \mathbb{F}_{q^m}^n \longrightarrow \mathcal{P}(\mathbb{F}_q^\mathbf{M}, \mathbf{n}) $ as follows. For $ \mathbf{c} = (\mathbf{c}^{(1)}, $ $ \mathbf{c}^{(2)}, $ $ \ldots, $ $ \mathbf{c}^{(\ell)}) \in \mathbb{F}_{q^m}^n $, where $ \mathbf{c}^{(i)} \in \mathbb{F}_{q^m}^{n_i} $, for $ i = 1,2, \ldots, \ell $, we define
$$ \mathcal{I}_{\Sigma}(\mathbf{c}) = \left( \Col \left( \begin{array}{cc}
M_{\mathcal{A}}(\mathbf{c}^{(i)}) \\
\hline
I_{n_i}
\end{array} \right) \right)_{i=1}^\ell. $$
For a block code $ \mathcal{C} \subseteq \mathbb{F}_{q^m}^n $, we define its lifting as the sum-constant-dimension code $ \mathcal{I}_{\Sigma}(\mathcal{C}) \subseteq \mathcal{P}(\mathbb{F}_q^\mathbf{M}, \mathbf{n}) $.
\end{definition}

We now observe that the sum-subspace distance between lifted codewords coincides with twice the sum-rank distance of the codewords themselves. This is a straightforward extension of \cite[Prop.~4]{error-control} and is left to the reader.

\begin{proposition} \label{prop relation subspace and sum-rank distance}
Given $ \mathbf{c}, \mathbf{d} \in \mathbb{F}_{q^m}^n $ and a block code $ \mathcal{C} \subseteq \mathbb{F}_{q^m}^n $, it holds that
\begin{equation*}
\begin{split}
 \dd_{SS}(\mathcal{I}_{\Sigma}(\mathbf{c}), \mathcal{I}_{\Sigma}(\mathbf{d})) & = 2 \dd_{SR}(\mathbf{c}, \mathbf{d}), \\
 \dd_{SS}(\mathcal{I}_{\Sigma}(\mathcal{C})) & = 2 \dd_{SR}(\mathcal{C}).
\end{split}
\end{equation*}
\end{proposition}

Equipped with all these tools, we may finally claim the near optimality of lifted linearized Reed-Solomon codes. We follow the lines of \cite[Prop.~5]{error-control}. For simplicity in the formulas, we assume that the lengths $ n_i $ are all equal.

\begin{theorem} \label{th near optimal coding schemes non-coherent}
Assume that $ n^\prime = n_1 = n_2 = \ldots = n_\ell $ and recall that $ M_1 = M_2 = \ldots = M_\ell = n^\prime + m $. Assume also that $ 1 \leq \ell \leq q-1 $ and $ 1 \leq n^\prime \leq m $. Let $ 1 \leq k \leq n $ and define $ d = n-k+1 $. 

Denote by $ R_1 $ the rate of the lifted linearized Reed-Solomon code $ \mathcal{I}_{\Sigma}(\mathcal{C}^\sigma_{L,k}(\boldsymbol{\mathcal{B}},\gamma)) $ (Definitions~\ref{def lin RS codes} and \ref{def lifting}) and by $ R_2 $, the maximum rate of a sum-constant-dimension code of type $ [\mathbf{M}, \mathbf{n}, N, d]_q $, running over all possible $ N \in \mathbb{R} $. It holds that
$$ \frac{R_2 - R_1}{R_2} < \frac{\ell}{k} \cdot \frac{2}{m \log_2(q)} . $$
\end{theorem}
\begin{proof}
Let $ \mathcal{C} $ be a sum-constant-dimension code of type $ [\mathbf{M}, \mathbf{n}, \log_q|\mathcal{C}|, d]_q $ and with rate $ R_2 $. Let the notation be as in Theorem~\ref{th sum subspace singleton bound} for $ \mathcal{C} $ and for an arbitrary partition $ d-1 = \delta_1 + \delta_2 + \cdots + \delta_\ell $. By \cite[Lemma~4]{errors-network}, we have that
\begin{equation}
{M_i - \delta_i \brack M_i - n_i}_q = {n^\prime + m - \delta_i \brack m}_q < 4 q^{m(n^\prime - \delta_i)},
\label{eq bound on gaussian coefficients}
\end{equation}
for $ i = 1,2, \ldots, \ell $. Therefore, by Theorem~\ref{th sum subspace singleton bound}, we have that
$$ |\mathcal{C}| < \prod_{i=1}^\ell (4 q^{m(n^\prime - \delta_i)}) = 4^\ell q^{m(n - d + 1)} = 4^\ell q^{mk}, $$
since $ \ell n^\prime = n $, $ \delta_1 + \delta_2 + \cdots + \delta_\ell = d-1 $, and $ n-d+1 = k $ by the assumptions and Proposition~\ref{prop relation subspace and sum-rank distance}. Thus, it holds that
$$ R_2 < \frac{\log_q(4) \ell + mk}{\ell (m + n^\prime) n^\prime}. $$
On the other hand, we have that
$$ R_2 \geq R_1 = \frac{mk}{\ell (m + n^\prime) n^\prime}. $$
Hence $ R_2 - R_1 < \log_q(4) / ((m+n^\prime) n^\prime) $, and we conclude that 
$$ \frac{R_2 - R_1}{R_2} < \frac{\log_q(4)}{(m + n^\prime) n^\prime} \cdot \frac{\ell}{k} \cdot \frac{(m+ n^\prime) n^\prime}{m}, $$
from which the desired bound follows easily. 
\end{proof}

Observe that $ k/\ell $ is the rate of information packets injected per shot. If we inject at least one packet of information per shot, then $ \ell / k \leq 1 $, and the previous upper bound is simply $ 2 / (m \log_2 (q) ) $, which is essentially that in \cite[Prop.~5]{error-control}, and is independent of the number of shots. Moreover, as shown in the proof of \cite[Lemma~4]{errors-network}, the factor $ 4 $ in the upper bound (\ref{eq bound on gaussian coefficients}) on $ q $-ary Gaussian coefficients can be refined to 
$$ h(q) = \prod_{j = 1}^\infty \frac{1}{1 - q^{-j}} < 4, $$
which decreases quickly as $ q $ increases. Thus if $ k/\ell \geq 1 $, and following the previous proof, we may also obtain the bound
$$ \frac{R_2 - R_1}{R_2} < \frac{\log_q(h(q))}{m}, $$
which is very small even for moderate values of $ q $ and $ m $.

\section{A Welch-Berlekamp Sum-rank Decoding Algorithm for Linearized Reed-Solomon Codes} \label{sec Berlekamp-Welch}

In this section, we show how to adapt Loidreau's version
\cite{decoding-loidreau} of the Welch-Berlekamp decoding algorithm
\cite{welch-berlekamp} to a sum-rank decoding algorithm with quadratic
complexity over $ \mathbb{F}_{q^m} $ for linearized Reed-Solomon codes
(Definition~\ref{def lin RS codes}).

Since working with linearized polynomials in the operators (\ref{eq definition
linear operator}) requires keeping track of both $ a $ (the conjugacy class)
and $ b $ (basis vectors in that conjugacy class), we give the algorithm for
the skew metric \cite[Def.~9]{linearizedRS} and skew Reed-Solomon codes
\cite[Def.~7]{skew-evaluation1}, where evaluation points need not be
partitioned into classes. Using tools from \cite[Sec.~3]{linearizedRS}, we will
see in Subsection~\ref{subsec skew metrics} that both algorithms can be
translated into each other after $ \mathcal{O}(n) $ multiplications in $
\mathbb{F}_{q^m} $. 

We note that a skew-metric decoding algorithm for skew Reed-Solomon codes was
also recently given in \cite{boucher-decoding}. However, this algorithm has
cubic complexity. More importantly, it 
is not translated to the sum-rank metric, so it is not applicable for reliability and security in multishot network coding. In particular, it does not handle
erasures, non-coherent communication or wire-tapper observations. Finally, observe that the step from
cubic complexity to quadratic complexity is a bigger jump than from the state of the art
(previous to \cite{boucher-decoding}) to a sum-rank decoding algorithm of cubic complexity, as a function of the number of shots $ \ell $
and in operations over $ \mathbb{F}_2 $ (see Table \ref{table comparison}).  

We recall also that the use of skew polynomials was proposed in \cite{zhang}
for secret sharing. This corresponds to using skew Reed-Solomon codes
\cite{skew-evaluation1} as MDS codes for reliability and security. However,
without transforming these codes into MSRD linearized Reed-Solomon codes
\cite{linearizedRS}, such schemes are not suitable for multishot network
coding.

The notation throughout this section is as in Subsection~\ref{subsec lin RS
codes and duals}.

\subsection{Skew Metrics and Skew Reed-Solomon Codes} \label{subsec skew metrics}

\begin{table*}[!t]
\caption{Maximum-distance evaluation codes using skew and linearized polynomials}
\label{table skew ev codes}
\centering
\begin{tabular}{c||c|c|c|c|c}
\hline
&&&&&\\[-0.8em]
Code & Metric & Type of polynomials evaluated & Evaluation points & Lengths & Field \\[1.6pt]
\hline\hline
&&&&&\\[-0.8em]
Reed-Solomon (RS) \cite{reed-solomon} & Hamming & $ \mathbb{F}_q[x; {\rm Id}] = \mathbb{F}_q[x] $ & Pair-wise distinct & $ n = \ell $ & $ \mathbb{F}_q $ \\[1.6pt]
\cline{1-4}
&&&&&\\[-0.8em]
Generalized RS \cite{delsarte} & Hamming & $ \{ \mathbb{F}_q[\mathcal{D}_{\gamma^{i-1}}] \}_{i=1}^\ell \equiv $ Multipliers + $ \mathbb{F}_q[x] $ & Non-zero + pair-wise dist. & $ \ell \leq q $ ($ n^\prime = 1 $) & (Base: $ \mathbb{F}_q $) \\[1.6pt]
\hline\hline
&&&&&\\[-0.8em]
Skew RS \cite{skew-evaluation1} & Skew & $ \mathbb{F}_{q^m}[x ; \sigma] $ & P-independent & $ n = \ell n^\prime $ & $ \mathbb{F}_{q^m} $ \\[1.6pt]
\cline{1-4}
&&&&&\\[-0.8em]
Linearized RS \cite{linearizedRS} & Sum-rank & $ \{ \mathbb{F}_{q^m}[\mathcal{D}_{\gamma^{i-1}}] \}_{i=1}^\ell $ & Lin. indep. + conjugacy & $ \ell < q $ \& $ n^\prime \leq m $ & (Base: $ \mathbb{F}_q $)  \\[1.6pt]
\hline\hline
&&&&&\\[-0.8em]
Skew RS in $ C(1) $ & Skew in $ C(1) $ & $ \mathbb{F}_{q^m}[x ; \sigma] \approx \{ \sum_j a_j x^{q^j - 1} \mid a_j \in \mathbb{F}_{q^m} \} $ & Lin. indep. conjugants & $ n = n^\prime $ & $ \mathbb{F}_{q^m} $ \\[1.6pt]
\cline{1-4}
&&&&&\\[-0.8em]
Gabidulin \cite{gabidulin} & Rank & $ \mathbb{F}_{q^m}[\mathcal{D}_1] = \{ \sum_j a_j x^{q^j} \mid a_j \in \mathbb{F}_{q^m} \} $ & Lin. independent & $ n^\prime \leq m $ ($ \ell = 1 $) & (Base: $ \mathbb{F}_q $) \\[1.6pt]
\hline
\end{tabular}
\end{table*}

Define the skew polynomial ring $ \mathbb{F}_{q^m}[x; \sigma] $ as the vector space over $ \mathbb{F}_{q^m} $ with basis $ \{ x^i \mid i \in \mathbb{N} \} $ and with product given by the rules $ x^i x^j = x^{i + j} $ and
\begin{equation}
xa = \sigma(a) x,
\label{eq product constant and variable}
\end{equation}
for all $ a \in \mathbb{F}_{q^m} $ and all $ i,j \in \mathbb{N} $. Define the degree of a non-zero $ F = \sum_{i \in \mathbb{N}} F_i x^i \in \mathbb{F}_{q^m}[x; \sigma] $, denoted by $ \deg(F) $, as the maximum $ i \in \mathbb{N} $ such that $ F_i \neq 0 $. We also define $ \deg (0) = \infty $. 

This ring was introduced with more generality in \cite{ore}. It is non-commutative and both a left and right Euclidean domain. We will see that evaluations of linearized polynomials as in Definition~\ref{def linearized operators} correspond to arithmetic evaluations of skew polynomials, as defined in \cite{lam, lam-leroy}.

\begin{definition} [\textbf{Evaluation \cite{lam, lam-leroy}}]
Given $ F \in \mathbb{F}_{q^m}[x; \sigma] $, we define its evaluation in $ a \in \mathbb{F}_{q^m} $ as the unique $ F(a) \in \mathbb{F}_{q^m} $ such that there exists $ G \in \mathbb{F}_{q^m}[x; \sigma] $ with
$$ F = G (x-a) + F(a). $$
Given a subset $ \Omega \subseteq \mathbb{F}_{q^m} $, we denote by $ \mathbb{F}_{q^m}^\Omega $ the set of functions $ f : \Omega \longrightarrow \mathbb{F}_{q^m} $. We then define the evaluation map over $ \Omega $ as the linear map
\begin{equation}
E_\Omega : \mathbb{F}_{q^m}[x; \sigma] \longrightarrow \mathbb{F}_{q^m}^\Omega,
\label{def evaluation map}
\end{equation}
where $ f = E_\Omega  (F) \in \mathbb{F}_{q^m}^\Omega $ is given by $ f(a) = F(a) $, for all $ a \in \Omega $ and for $ F \in \mathbb{F}_{q^m}[x; \sigma] $. Again, we write $ E_\Omega^\sigma $ when there can be confusion about $ \sigma $.
\end{definition}

We will need some basic concepts regarding the zero sets of skew polynomials.

\begin{definition} [\textbf{Zeros of skew polynomials}]
Given a set $ A \subseteq \mathbb{F}_{q^m}[x;\sigma] $, we define its zero set as
$$ Z(A) = \{ a \in \mathbb{F}_{q^m} \mid F(a) = 0, \forall F \in A \}. $$
Given a subset $ \Omega \subseteq \mathbb{F}_{q^m} $, we define its associated ideal as
$$ I(\Omega) = \{ F \in \mathbb{F}_{q^m}[x;\sigma] \mid F(a) = 0, \forall a \in \Omega \}. $$
\end{definition}

Observe that $ I(\Omega) $ is a left ideal in $ \mathbb{F}_{q^m}[x;\sigma] $. Since $ \mathbb{F}_{q^m}[x;\sigma] $ is a right Euclidean domain, there exists a unique monic skew polynomial $ F_\Omega \in I(\Omega) $ of minimal degree among those in $ I(\Omega) $, which in turn generates $ I(\Omega) $ as left ideal. Such a skew polynomial is called the \textit{minimal skew polynomial} of $ \Omega $ \cite{lam, lam-leroy}. With this in mind, we may recall the concepts of \textit{P-closed sets}, \textit{P-independence} and \textit{P-bases} from \cite[Sec.~4]{algebraic-conjugacy} (see also \cite{lam}):

\begin{definition} [\textbf{P-bases \cite{lam, algebraic-conjugacy}}]
Given a subset $ \Omega \subseteq \mathbb{F}_{q^m} $, we define its P-closure as $ \overline{\Omega} = Z(I(\Omega)) = Z(F_\Omega) $, and we say that it is P-closed if $ \overline{\Omega} = \Omega $.

We say that $ a \in \mathbb{F}_{q^m} $ is P-independent from $ \Omega \subseteq \mathbb{F}_{q^m} $ if it does not belong to $ \overline{\Omega} $. A set $ \Omega \subseteq \mathbb{F}_{q^m} $ is called P-independent if every $ a \in \Omega $ is P-independent from $ \Omega \setminus \{ a \} $. 

Given a P-closed set $ \Omega \subseteq \mathbb{F}_{q^m} $, we say that $ \mathcal{B} \subseteq \Omega $ is a P-basis of $ \Omega $ if it is P-independent and $ \Omega = \overline{\mathcal{B}} $.
\end{definition}

The following two lemmas give simple and useful connection between P-bases and minimal skew polynomials.

\begin{lemma} [\textbf{\cite{lam}}] \label{lemma degree of min skew pol}
Given a finite set $ \Omega \subseteq \mathbb{F}_{q^m} $, it holds that
$$ \deg(F_\Omega) \leq |\Omega|, $$
where equality holds if, and only if, $ \Omega $ is P-independent.
\end{lemma}

\begin{lemma} [\textbf{\cite{lam, algebraic-conjugacy}}] \label{lemma any two basis same size}
Given a P-closed set $ \Omega \subseteq \mathbb{F}_{q^m} $, it admits a P-basis and any two of them have the same number of elements, which is 
$$ \Rk(\Omega) \stackrel{ \textrm{def} }{=} \deg(F_\Omega). $$
\end{lemma}

From now on, we will fix a P-closed set $ \Omega \subseteq \mathbb{F}_{q^m} $ with $ n = \Rk(\Omega) $. We will also denote by $ \mathbb{F}_{q^m}[x; \sigma]_n $ the vector space of skew polynomials of degree less than $ n $.

The following lemma, given more generally in \cite[Th.~8]{lam}, is the main idea behind skew weights and skew Reed-Solomon codes. 

\begin{lemma}[\textbf{Lagrange interpolation \cite{lam}}] \label{lemma lagrange interpolation}
The evaluation map (\ref{def evaluation map}) restricted to $ \mathbb{F}_{q^m}[x; \sigma]_n $, that is
$$ E_\mathcal{B} : \mathbb{F}_{q^m}[x; \sigma]_n \longrightarrow \mathbb{F}_{q^m}^\mathcal{B}, $$
is a vector space isomorphism, for any P-basis $ \mathcal{B} $ of $ \Omega $.
\end{lemma}

The definitions of skew weights \cite[Def.~9]{linearizedRS} and skew Reed-Solomon codes \cite[Def.~7]{skew-evaluation1} can now be given as follows:

\begin{definition} [\textbf{Skew weights \cite{linearizedRS}}] \label{def skew weights}
Given $ F \in \mathbb{F}_{q^m}[x; \sigma]_n $ and $ f = E_\mathcal{B}(F) \in \mathbb{F}_{q^m}^\mathcal{B} $, for a P-basis $ \mathcal{B} $ of $ \Omega $, we define their skew weight over $ \Omega $ as
$$ \wt_\mathcal{B}(f) = \wt_\Omega(F) = n - \Rk(Z_\Omega(F)), $$
where $ Z_\Omega(F) = Z(F) \cap \Omega = Z(\{ F, F_\Omega \}) $ is the P-closed set of zeros of $ F $ in $ \Omega $. 
\end{definition}

As shown in \cite{linearizedRS}, these functions are indeed weights and define a corresponding metric, called the skew metric.

\begin{definition} [\textbf{Skew Reed-Solomon codes \cite{skew-evaluation1}}] \label{def skew RS codes}
For each $ k = 0,1,2, \ldots, n $, we define the ($ k $-dimensional) skew Reed-Solomon code over a P-basis $ \mathcal{B} $ of $ \Omega $ as the linear code
$$ \mathcal{C}_{\mathcal{B},k}^{\sigma} = E_\mathcal{B}^{\sigma}(\mathbb{F}_{q^m}[x; \sigma]_k) \subseteq \mathbb{F}_{q^m}^\mathcal{B} . $$
\end{definition}

The exact connection between skew metrics and sum-rank metrics, and between
skew Reed-Solomon codes and linearized Reed-Solomon codes was given in
\cite[Sec.~3]{linearizedRS}. We summarize it in the next theorem, where the
first claim on $ \mathcal{B} $ combines \cite[Th.~23]{lam} and
\cite[Th.~4.5]{lam-leroy}.

\begin{theorem} [\textbf{\cite{lam, lam-leroy, linearizedRS}}] \label{th connection skew and lin polynomials}
Assume that $ 1 \leq \ell \leq q-1 $ and $ 1 \leq n_i \leq m $, for $ i = 1,2, \ldots, \ell $. Fix linearly independent sets $ \mathcal{B}^{(i)} = \{ \beta_1^{(i)}, \beta_2^{(i)}, \ldots, \beta_{n_i}^{(i)} \} \subseteq \mathbb{F}_{q^m} $ over $ \mathbb{F}_q $, and a primitive element $ \gamma $ of $ \mathbb{F}_{q^m} $. With notation as in Subsection~\ref{subsec lin RS codes and duals}, define
\begin{equation}
b_j^{(i)} = \mathcal{D}_{\gamma^{i-1}}(\beta_j^{(i)})(\beta_j^{(i)})^{-1},
\label{eq transforming bases to P-bases}
\end{equation}
for $ j = 1,2, \ldots, n_i $ and $ i = 1,2,\ldots, \ell $. Then $ \mathcal{B} = \{ b_j^{(i)} \mid 1 \leq j \leq n_i, 1 \leq i \leq \ell \} $ is a P-basis of $ \Omega = \overline{\mathcal{B}} $. 

Next, if $ F = \sum_{i \in \mathbb{N}} F_i x^i \in \mathbb{F}_{q^m}[x ; \sigma] $ and $ F^{\mathcal{D}_a} = \sum_{i \in \mathbb{N}} F_i \mathcal{D}_a^i $, for some $ a \in \mathbb{F}_{q^m} $, then 
$$ F(\mathcal{D}_a(\beta) \beta^{-1}) = F^{\mathcal{D}_a}(\beta)\beta^{-1}, $$
for all $ \beta \in \mathbb{F}_{q^m}^* $. Define now the linear map $ \psi_\mathcal{B} : \mathbb{F}_{q^m}^n \longrightarrow \mathbb{F}_{q^m}^\mathcal{B} $ by $ \psi_\mathcal{B} (\mathbf{c}^{(1)}, $ $ \mathbf{c}^{(2)}, $ $ \ldots, $ $ \mathbf{c}^{(\ell)}) = f $, where 
\begin{equation}
f(b_j^{(i)}) = c_j^{(i)} (\beta_j^{(i)})^{-1},
\label{eq transforming the received word}
\end{equation}
for $ j = 1,2, \ldots, n_i $ and $ i = 1,2,\ldots, \ell $. It holds that
$$ \psi_\mathcal{B}(\mathcal{C}^\sigma_{L,k}) = \mathcal{C}^\sigma_{\mathcal{B},k}, $$
for all $ k = 0,1,\ldots,n $. For $ \mathbf{c} \in \mathbb{F}_{q^m}^n $, we also have that
$$ \wt_\mathcal{B}(\psi_\mathcal{B}(\mathbf{c})) = \wt_{SR}(\mathbf{c}). $$ 
\end{theorem}

What this theorem implies is that sum-rank metrics and linearized Reed-Solomon
codes can be treated as skew metrics and skew Reed-Solomon codes. We only need
to multiply the received word (a vector in $ \mathbb{F}_{q^m}^n $)
coordinate-wise by $ n $ elements in $ \mathbb{F}_{q^m} $ as in (\ref{eq
transforming the received word}), and similarly compute the corresponding
P-basis as in (\ref{eq transforming bases to P-bases}). This requires $
\mathcal{O}(n) $ multiplications in $ \mathbb{F}_{q^m} $. Thus, decoding
algorithms
can be directly translated from one scenario to the other at the expense
of just 
$\mathcal{O}(n) $ multiplications in $ \mathbb{F}_{q^m} $. 

Theorem \ref{th connection skew and lin polynomials} also implies
that linearized Reed-Solomon codes can be seen as generalized skew Reed-Solomon codes
\cite[Def. 9]{skew-evaluation2} for a special choice of column multipliers. 

Table \ref{table skew ev codes} provides a summary of
skew and linearized Reed-Solomon codes, and how they interpolate the intermediate
cases between Reed-Solomon and Gabidulin codes. In that table, $ C(1) $ stands for 
the conjugacy class of $ 1 \in \mathbb{F}_{q^m} $ over $ \mathbb{F}_q $. The 
skew Reed-Solomon codes corresponding to Gabidulin codes are not given a name or 
reference in the table, since they are only considered as a particular example in the literature.

\subsection{Key Equations} \label{subsec key equations}

In this subsection, we will present the skew polynomial version of the key equations in the Welch-Berlekamp algorithm. 

Fix a P-closed set $ \Omega \subseteq \mathbb{F}_{q^m} $, one of its P-bases $ \mathcal{B} = \{ b_1, $ $ b_2, $ $ \ldots, b_n \} $ and $ 1 \leq k \leq n $, where $ n = |\mathcal{B}| = \Rk(\Omega) $. We would like to decode up to 
$$ t \stackrel{def}{=} \left\lfloor \frac{n-k}{2} \right\rfloor $$ 
skew errors for the skew Reed-Solomon code $ \mathcal{C} = \mathcal{C}^\sigma_{\mathcal{B},k} \subseteq \mathbb{F}_{q^m}^\mathcal{B} $ (see Definitions~\ref{def skew weights} and \ref{def skew RS codes}). Assume then that $ f \in \mathcal{C} $, $ e \in \mathbb{F}_{q^m}^\mathcal{B} $ is such that $ \wt_\mathcal{B}(e) \leq t $, and $ r = f + e $ is the received word. By Lemma~\ref{lemma lagrange interpolation}, we may instead consider 
$$ R = F + G \in \mathbb{F}_{q^m}[x;\sigma]_n, $$ 
where $ F \in \mathbb{F}_{q^m}[x;\sigma]_k $, $ G \in \mathbb{F}_{q^m}[x;\sigma]_n $ with $ \wt_\Omega(G) \leq t $, $ f = E_\mathcal{B}(F) $ and $ e = E_\mathcal{B}(G) $. Following the original idea of the Welch-Berlekamp algorithm, we want to find a non-zero skew polynomial $ L \in \mathbb{F}_{q^m}[x;\sigma] $ of degree $ \deg(L) \leq t $, such that
\begin{equation}
(LR)(b_i) = (LF)(b_i),
\label{key equation 1}
\end{equation}
for $ i = 1,2, \ldots, n $. But since we do not know $ F $ (we want to find $ F $), we look instead for skew polynomials $ L,Q \in \mathbb{F}_{q^m}[x;\sigma] $ of degrees $ \deg(L) \leq t $ and $ \deg(Q) \leq t + k - 1 $, such that
\begin{equation}
(LR)(b_i) = Q(b_i),
\label{key equation 2}
\end{equation}
for $ i = 1,2, \ldots, n $. To solve these equations, the first problem that we encounter is evaluating a product of two skew polynomials. This is solved by considering the product rule \cite[Th.~2.7]{lam-leroy}.

\begin{lemma} [\textbf{Product rule \cite{lam-leroy}}] \label{lemma product rule}
Given $ U,V \in \mathbb{F}_{q^m}[x;\sigma] $ and $ a \in \mathbb{F}_{q^m} $, let $ c = V(a) $. If $ c = 0 $, then $ (UV)(a) = 0 $, and if $ c \neq 0 $, then
$$ (UV)(a) = U(a^c) V(a), $$
where we use the notation $ a^c = \sigma(c)c^{-1}a $.
\end{lemma}

Therefore, if $ r_i = R(b_i) $, we need to find $ L $ and $ Q $ as before, satisfying $ Q(b_i) = 0 $ for $ i $ such that $ r_i = 0 $, and
\begin{equation}
L(b_i^{r_i}) r_i = Q(b_i),
\label{key equation 3}
\end{equation}
for $ i $ such that $ r_i \neq 0 $. Observe that $ L(b_i^{r_i})r_i = L^\mathcal{D}(r_i) $ by Theorem~\ref{th connection skew and lin polynomials}, for $ \mathcal{D} = \mathcal{D}_{b_i} $, which is also defined for $ r_i = 0 $. Hence it makes sense to define ``$ L(b_i^{r_i})r_i = 0 $'' when $ r_i = 0 $, and we may consider (\ref{key equation 3}), for all $ i = 1,2, \ldots, n $.

We start by checking that (\ref{key equation 1}) can be solved (hence (\ref{key equation 3}) can also be solved). The next two propositions can also be proven by results from \cite{boucher-decoding} (see Remark \ref{remark connection with boucher}).

\begin{proposition} \label{prop key eqs can be solved}
There exists a non-zero $ L \in \mathbb{F}_{q^m}[x;\sigma] $ of degree $ \deg(L) \leq t $ satisfying (\ref{key equation 1}), for $ i = 1,2, \ldots, n $ (recall that $ \wt_\Omega(G) = \wt_\Omega(R - F) \leq t $). 
\end{proposition}
\begin{proof}
Define $ e_i = G(b_i) $, for $ i = 1,2,\ldots, n $. By Lemma~\ref{lemma product rule}, we see that (\ref{key equation 1}) is satisfied if
\begin{equation}
L(b_i^{e_i}) = 0,
\label{eq proof for existence key eqs 1}
\end{equation}
for $ i $ such that $ e_i \neq 0 $. The result \cite[Prop.~14]{linearizedRS} says that there exists a P-basis $ \mathcal{A} = \{ a_1, a_2, \ldots, a_n \} $ of $ \Omega $ such that 
$$ \wt_H(E_\mathcal{A}(G)) = \wt_\Omega(G) \leq t, $$
where $ \wt_H $ denotes Hamming weight. Let $ d_i = G(a_i) $, for $ i = 1, 2, \ldots, n $, and let $ \Delta = \{ a_i^{d_i} \mid d_i \neq 0 \} $. Since $ |\Delta| \leq t $, then $ \deg(F_\Delta) \leq t $ by Lemma~\ref{lemma degree of min skew pol}. Thus by choosing $ L = F_\Delta $, we see that
$$ L(a_i^{d_i}) = 0, $$
whenever $ d_i \neq 0 $. By Lemma~\ref{lemma product rule}, this is equivalent to 
$$ (LG)(a_i) = 0, $$
for $ i = 1,2,\ldots, n $. In other words, $ E_\mathcal{A}(LG) = 0 $, which by Lemma~\ref{lemma lagrange interpolation} is equivalent to $ E_\mathcal{B}(LG) = 0 $. Now again by Lemma~\ref{lemma product rule}, this means that (\ref{eq proof for existence key eqs 1}) is satisfied, and we are done.
\end{proof}

The second ingredient is checking that by solving (\ref{key equation 3}), we may recover the message skew polynomial $ F $ by a right Euclidean division.

\begin{proposition} \label{prop key 3 sol gives message pol}
Assume that $ L,Q \in \mathbb{F}_{q^m}[x;\sigma] $ satisfy (\ref{key equation 3}), with $ \deg(L) \leq t $ and $ \deg(Q) \leq t + k - 1 $. Then it holds that
$$ Q = LF. $$
\end{proposition}
\begin{proof}
First, define $ c_i = (LF - Q)(b_i) $, for $ i = 1,2, \ldots, n $. Observe that if $ a \in \Omega $ and $ F(a) = R(a) $, then $ (LF - Q)(a) = 0 $ by hypothesis and Lemma~\ref{lemma product rule}. Hence
$$ \wt_\Omega(LF - Q) \leq \wt_\Omega(F-R) \leq t. $$ 
Therefore applying Proposition~\ref{prop key eqs can be solved} to $ LF $ and $ Q $ instead of $ F $ and $ R $, we see that there exists a non-zero $ L_0 \in \mathbb{F}_{q^m}[x;\sigma] $ of degree $ \deg(L_0) \leq t $ satisfying
$$ L_0(b_i^{c_i}) = 0, $$
for $ i $ such that $ c_i \neq 0 $. Define $ P = L_0(LF - Q) $. If $ c_i = 0 $, then by Lemma~\ref{lemma product rule}, we have that $ P(b_i) = 0 $. Now if $ c_i \neq 0 $, then by Lemma~\ref{lemma product rule}, we have that
$$ P(b_i) = L_0(b_i^{c_i}) = 0. $$
In other words, $ E_\mathcal{B}(P) = 0 $, and since $ \deg(P) < n $ (here is where we use that $ t = \lfloor \frac{n-k}{2} \rfloor $), then $ P = 0 $ by Lemma~\ref{lemma lagrange interpolation}. The result follows since $ \mathbb{F}_{q^m}[x;\sigma] $ is an integral domain.
\end{proof}

Thus we will be able to efficiently decode if we can solve (\ref{key equation 3}), since then we may just perform Euclidean division, whose complexity is of $ t(t+k-1) = \mathcal{O}(n^2) $ multiplications in $ \mathbb{F}_{q^m} $. Equations (\ref{key equation 3}) form a system of $ n $ linear equations whose unknowns are the coefficients of $ L $ and $ Q $. However, solving such a system by Gaussian elimination has complexity $ \mathcal{O}(n^3) $ over $ \mathbb{F}_{q^m} $. This is the approach in \cite{boucher-decoding}. In the next subsection, we see how to reduce it to $ \mathcal{O}(n^2) $.

\begin{remark} \label{remark connection with boucher}
Propositions \ref{prop key eqs can be solved} and \ref{prop key 3 sol gives message pol} can also be derived by certain combination of results from \cite{boucher-decoding}. However, the machinery developed in \cite{boucher-decoding} rewrites skew weights \cite[Def. 9]{linearizedRS} in terms of least common multiples instead of P-closed sets and P-independence (Subsec. \ref{subsec skew metrics}), being the latter essential to connect skew weights with sum-rank weights (see Theorem \ref{th connection skew and lin polynomials}). Thus we have preferred to keep a self-contained proof based on the machinery in Subsec. \ref{subsec skew metrics}.
\end{remark}

\subsection{Algorithm with Quadratic Complexity} \label{subsec quadratic reduction}

In this subsection, we follow the steps in \cite[Subsec.~5.2]{decoding-loidreau} to solve (\ref{key equation 3}) with overall complexity of $ \mathcal{O}(n^2) $ multiplications in $ \mathbb{F}_{q^m} $. The idea is to construct sequences of skew polynomials $ L_j, Q_j, \widetilde{L}_j, \widetilde{Q}_j \in \mathbb{F}_{q^m}[x; \sigma] $ such that
\begin{equation}
\left. \begin{array}{rl}
L_j(b_i^{r_i}) r_i - Q_j(b_i) & = 0 \\
\widetilde{L}_j(b_i^{r_i}) r_i - \widetilde{Q}_j(b_i) & = 0
\end{array} \right\rbrace
\label{eq key eqs for sequence}
\end{equation}
for $ i = 1,2, \ldots, j $ and $ j = k,k+1, \ldots, n $, and where $ \deg(L_n), $ $ \deg(\widetilde{L}_n) \leq t $ and $ \deg(Q_n), $ $ \deg(\widetilde{Q}_n) \leq t + k - 1 $. If we have constructed the $ j $th skew polynomials, we define
\begin{equation}
\begin{array}{rl}
s_j = & L_j(b_{j+1}^{r_{j+1}}) r_{j+1} - Q_j(b_{j+1}), \\
\widetilde{s}_j = & \widetilde{L}_j(b_{j+1}^{r_{j+1}}) r_{j+1} - \widetilde{Q}_j(b_{j+1}).
\end{array}
\label{eq def seq of s}
\end{equation}
Next we define
\begin{equation}
\begin{array}{rl}
L_{j+1} = & (x - b_{j+1}^{s_j}) L_j, \\
Q_{j+1} = & (x - b_{j+1}^{s_j}) Q_j,
\end{array}
\label{eq def of first sequence}
\end{equation}
if $ s_j \neq 0 $; and we also define $ L_{j+1} = L_j $ and $ Q_{j+1} = Q_j $, if $ s_j = 0 $. We also define
\begin{equation}
\begin{array}{rl}
\widetilde{L}_{j+1} = & s_j \widetilde{L}_j - \widetilde{s}_j L_j, \\
\widetilde{Q}_{j+1} = & s_j \widetilde{Q}_j - \widetilde{s}_j Q_j.
\end{array}
\label{eq def of second sequence}
\end{equation}

The important fact is that then the $ (j+1) $th skew polynomials satisfy (\ref{eq key eqs for sequence}), as we now show.

\begin{proposition}
If the $ j $th skew polynomials satisfy (\ref{eq key eqs for sequence}) for $ i = 1,2, \ldots, j $, then the $ (j+1) $th skew polynomials given in (\ref{eq def of first sequence}) and (\ref{eq def of second sequence}) satisfy (\ref{eq key eqs for sequence}) for $ i = 1,2, \ldots, j, j+1 $.
\end{proposition}
\begin{proof}
In all cases, if $ i \leq j $, then the result follows from Lemma~\ref{lemma product rule}, since (\ref{eq key eqs for sequence}) is satisfied for the $ j $th skew polynomials. Hence we only need to check the case $ i = j+1 $. It is straightforward to check it for the skew polynomials given in (\ref{eq def of second sequence}), since evaluation is a linear map. Hence we only need to prove it for the skew polynomials given in (\ref{eq def of first sequence}), and only in the case $ s_j \neq 0 $. Denote
$$ \begin{array}{rl}
c = & L_j(b_{j+1}^{r_{j+1}}) r_{j+1}, \\
d = & Q_j(b_{j+1}).
\end{array} $$
Observe that $ s_j = c - d \neq 0 $. We have that
$$ L_{j+1}(b_{j+1}^{r_{j+1}}) r_{j+1} - Q_{j+1}(b_{j+1}) $$
$$ = (b_{j+1}^c - b_{j+1}^{s_j})c - (b_{j+1}^d - b_{j+1}^{s_j})d $$
$$ = (b_{j+1}^c c - b_{j+1}^d d) - b_{j+1}^{c-d}(c-d) = 0, $$
where the first equality follows from Lemma~\ref{lemma product rule}, and the last equality follows from the fact that the map $ \lambda \mapsto F(b^\lambda) \lambda = F^{\mathcal{D}_b}(\lambda) $ is additive, for $ b, \lambda \in \mathbb{F}_{q^m} $ (see Theorem~\ref{th connection skew and lin polynomials}).
\end{proof}

We now describe the actual steps of the algorithm, following the same scheme as in \cite[Table~1]{decoding-loidreau}. First we need to precompute the minimum skew polynomial $ F_k = F_{\{ b_1, b_2, \ldots, b_k \}} $ and the unique Lagrange interpolating skew polynomial $ G \in \mathbb{F}_{q^m}[x; \sigma]_k $ such that $ G(b_i) = r_i $, for $ i = 1,2, \ldots, k $ (Lemma~\ref{lemma lagrange interpolation}). Both can be computed with quadratic complexity by Newton interpolation (see Appendix~\ref{app newton interpolation}). 

The algorithm is as follows, where the inputs are $ k \in \mathbb{N} $, $ (b_1, b_2, \ldots, b_n) \in \mathbb{F}_{q^m}^n $ and $ (r_1, r_2, \ldots, r_n) \in \mathbb{F}_{q^m}^n $.
\begin{enumerate}
\item
Initialization and Newton interpolation:
\begin{itemize}
\item
Compute $ F_k $ and $ G $.
\item
Set $ L_k = 0 $ and $ \widetilde{L}_k = x $.
\item
Set $ Q_k = F_k $ and $ \widetilde{Q}_k = G $.
\end{itemize}
\item
Alternating step: For each $ j = k,k+1,\ldots, n-1 $, do
\begin{itemize}
\item
Compute $ s_j $ and $ \widetilde{s}_j $ as in (\ref{eq def seq of s}).
\item
Exchange the values $ L_j \longleftrightarrow \widetilde{L}_j $, $ Q_j \longleftrightarrow \widetilde{Q}_j $ and $ s_j \longleftrightarrow \widetilde{s}_j $.
\item
Compute $ L_{j+1} $, $ Q_{j+1} $, $ \widetilde{L}_{j+1} $ and $ \widetilde{Q}_{j+1} $ as in (\ref{eq def of first sequence}) and (\ref{eq def of second sequence}).
\end{itemize}
\item
Euclidean division: 
\begin{itemize}
\item
Set $ L = \widetilde{L}_n $ and $ Q = \widetilde{Q}_n $.
\item
Compute $ F $ such that $ Q = LF $ by Euclidean division.
\end{itemize}
\item
Return the coefficients of $ F $: $ (F_0, F_1, \ldots, F_{k-1}) \in \mathbb{F}_{q^m}^k $.
\end{enumerate}

As observed in \cite{decoding-loidreau}, one can save memory by assigning each update of the sequences of skew polynomials to themselves.

The degrees of $ L $ and $ Q $ are upper bounded by $ t $ and $ t+k-1 $, respectively. This can be shown exactly as in \cite[Subsec.~5.2]{decoding-loidreau}. The idea is that we increase the degree at most by $ 1 $ for one of the two sequences in every step. By exchanging both sequences in each step and taking $ n-k $ steps, at the end we have that
$$ \begin{array}{rl}
\deg(L) \leq & \lfloor \frac{n-k}{2} \rfloor, \\
\deg(Q) \leq & \lfloor \frac{n-k}{2} \rfloor + k - 1.
\end{array} $$

\subsection{Overall Complexity} \label{subsec complexity algorithm}

There are mainly three basic operations in the arithmetic of skew and linearized polynomials, namely, multiplications and additions in $ \mathbb{F}_{q^m} $, and applying $ \sigma^j(a) = a^{q^j} $, for $ a \in \mathbb{F}_{q^m} $ and $ 1 \leq j < m $. Multiplications are more expensive than additions, but we will count both. We anticipate that, in $ \mathcal{O} $ notation, the amounts of multiplications and additions in our algorithm are roughly the same. We will neglect how many times we apply $ \sigma^j $ since, when representing elements in $ \mathbb{F}_{q^m} $ as vectors over $ \mathbb{F}_q $ using a normal basis, applying $ \sigma^j $ amounts to a cyclic shift of coordinates. 

First, Newton interpolation in $ k $ P-independent points requires $ \mathcal{O} (k^2) $ multiplications and $ \mathcal{O}(k^2) $ additions (see Appendix~\ref{app newton interpolation}). This is then the complexity of Step 1.

Second, evaluating a skew polynomial of degree $ d $ by Horner's rule (i.e., Euclidean division by $ x-a $) requires $ \mathcal{O} (d) $ multiplications and $ \mathcal{O} (d) $ additions. Hence Step 2 requires $ \mathcal{O}((n-k)(t+k-1)) $ multiplications and $ \mathcal{O}((n-k)(t+k-1)) $ additions since, in each of the $ n - k $ iterations of Step 2, we compute $ 4 $ evaluations of skew polynomials of degree at most $ t+k-1 $ to compute $ s_j $ and $ \widetilde{s}_j $, plus another $ \mathcal{O}(t+k-1) $ multiplications and $ \mathcal{O}(t+k-1) $ additions to compute $ L_{j+1} $, $ Q_{j+1} $, $ \widetilde{L}_{j+1} $ and $ \widetilde{Q}_{j+1} $.

Third, Euclidean division of a skew polynomial of degree $ t+k-1 $ by another one of degree $ t $ requires $ \mathcal{O}(t(t+k-1)) $ multiplications and $ \mathcal{O}(t(t+k-1)) $ additions. This is the complexity of Step 3.

Summing all these numbers, we see that the previous algorithm has complexity of $ \mathcal{O}(t(t+k-1) + k^2) $ multiplications and $ \mathcal{O}(t(t+k-1) + k^2) $ additions in $ \mathbb{F}_{q^m} $, which can be simplified to $ \mathcal{O}(n^2) $ since $ 2t \leq n $ and $ k \leq n $.

\subsection{Translating the Algorithm into Sum-rank Decoding of Linearized Reed-Solomon Codes} \label{subsec translation into sum-rank algorithm}

To transform this algorithm into a sum-rank decoding algorithm for linearized Reed-Solomon codes, we proceed as follows. First we build the P-basis $ \mathcal{B} = \{ b_1, b_2, \ldots, b_n \} $ by (\ref{eq transforming bases to P-bases}) from the linearly independent sets $ \mathcal{B}^{(i)} = \{ \beta_1^{(i)}, $ $ \beta_2^{(i)}, $ $ \ldots, $ $ \beta_{n_i}^{(i)} \} \subseteq \mathbb{F}_{q^m} $ over $ \mathbb{F}_q $, $ i = 1,2, \ldots, \ell $, and the primitive element $ \gamma \in \mathbb{F}_{q^m}^* $. In other words, we compute
\begin{equation}
b_j^{(i)} = \sigma(\beta_j^{(i)}) (\beta_j^{(i)})^{-1} \gamma^{i-1}, 
\label{eq transforming basis to P-basis}
\end{equation}
where $ b_j^{(i)} = b_l $ and $ l = n_1 + n_2 + \cdots + n_{i-1} + j $, for $ j = 1, $ $2, $ $ \ldots, n_i $ and $ i = 1,2, \ldots, \ell $. Then we map the received word to the vector $ (r_1, r_2, \ldots, r_n) \in \mathbb{F}_{q^m}^n $ by the map $ \psi_\mathcal{B} $ in Theorem~\ref{th connection skew and lin polynomials}, or equivalently, by (\ref{eq transforming the received word}). In other words, we compute
\begin{equation}
r_j^{(i)} = y_j^{(i)} (\beta_j^{(i)})^{-1}, 
\label{eq transforming received words}
\end{equation}
where $ r_j^{(i)} = r_l $ and $ l = n_1 + n_2 + \cdots + n_{i-1} + j $, for $ j = 1, $ $2, $ $ \ldots, n_i $, and where $ \mathbf{y}^{(i)} = (y_1^{(i)}, y_2^{(i)}, \ldots, y_{n_i}^{(i)}) \in \mathbb{F}_{q^m}^{n_i} $ is the received word in the $ i $th shot, for $ i = 1,2, \ldots, \ell $. 

These two processes require $ 3n $ multiplications in $ \mathbb{F}_{q^m} $. Finally we run the algorithm in Subsection~\ref{subsec quadratic reduction}. By Theorem~\ref{th connection skew and lin polynomials}, the coefficients of $ F $ coincide with those of $ F^{\mathcal{D}_a} $, for any $ a \in \mathbb{F}_{q^m} $, and also give the message in $ \mathbb{F}_{q^m}^k $ for the corresponding linearized Reed-Solomon code.

\subsection{Including Erasures for Coherent Communication} \label{subsec decoding for coherent}

Consider that we have used the linearized Reed-Solomon code $ \mathcal{C}^\sigma_{L,k}(\boldsymbol{\mathcal{B}},\gamma) \subseteq \mathbb{F}_{q^m}^n $ as in Definition~\ref{def lin RS codes}, and we receive
$$ \mathbf{y} = \mathbf{c} A^T + \mathbf{e} \in \mathbb{F}_{q^m}^{n-\rho}, $$
where $ \wt_{SR}(\mathbf{e}) \leq t $, $ A = \diag(A_1, A_2, \ldots, A_n) \in \mathbb{F}_q^{(n-\rho) \times n} $, and $ A_i \in \mathbb{F}_q^{(n_i-\rho_i) \times n_i} $, for $ i = 1,2, \ldots, \ell $. First, we may assume that all $ A_i $ have full rank. Next compute 
$$ (\alpha_1^{(i)}, \alpha_2^{(i)}, \ldots, \alpha_{n_i-\rho_i}^{(i)}) = (\beta_1^{(i)}, \beta_2^{(i)}, \ldots, \beta_{n_i}^{(i)}) A_i^T \in \mathbb{F}_{q^m}^{n_i-\rho_i}, $$
where $ \mathcal{B}^{(i)} = \{ \beta_1^{(i)}, \beta_2^{(i)}, \ldots, \beta_{n_i}^{(i)} \} $, and define $ \mathcal{A}^{(i)} = \{ \alpha_1^{(i)}, $ $ \alpha_2^{(i)}, $ $ \ldots, $ $ \alpha_{n_i-\rho_i}^{(i)} \} $, for $ i = 1,2, \ldots, \ell $. It follows that the sets $ \mathcal{A}_i $ are linearly independent over $ \mathbb{F}_q $. Furthermore, by the linearity over $ \mathbb{F}_q $ of the operators in Definition~\ref{def linearized operators}, we have that
$$ (\mathcal{D}_{\gamma^{i-1}}^j(\alpha_1^{(i)}), \mathcal{D}_{\gamma^{i-1}}^j(\alpha_2^{(i)}), \ldots, \mathcal{D}_{\gamma^{i-1}}^j(\alpha_{n_i-\rho_i}^{(i)})) $$
$$ = (\mathcal{D}_{\gamma^{i-1}}^j(\beta_1^{(i)}), \mathcal{D}_{\gamma^{i-1}}^j(\beta_2^{(i)}), \ldots, \mathcal{D}_{\gamma^{i-1}}^j(\beta_{n_i}^{(i)})) A_i^T $$for $ j = 0,1, \ldots, k-1 $ and $ i = 1,2, \ldots, \ell $. Thus $ \mathbf{c} A^T $ corresponds to the evaluation codeword in the linearized Reed-Solomon code $ \mathcal{C}^\sigma_{L,k}(\boldsymbol{\mathcal{A}},\gamma) \subseteq \mathbb{F}_{q^m}^{n-\rho} $, where $ \boldsymbol{\mathcal{A}} = (\mathcal{A}^{(1)}, $ $ \mathcal{A}^{(2)}, $ $ \ldots, $ $ \mathcal{A}^{(\ell)}) $. Moreover, the number of sum-rank errors is at most $ t = \lfloor \frac{n-\rho-k}{2} \rfloor $, where $ n-\rho $ is the length of the new code. 

In conclusion, to recover the message we only need to compute $ \boldsymbol{\mathcal{A}} $ to find the new code $ \mathcal{C}^\sigma_{L,k}(\boldsymbol{\mathcal{A}},\gamma) $, and then run the algorithm in Subsection~\ref{subsec quadratic reduction}. Such a computation is equivalent to the multiplication of a vector in $ \mathbb{F}_{q^m}^n $ with a matrix in $ \mathbb{F}_q^{n \times (n-\rho)} $, thus about $ \mathcal{O}(n(n-\rho)) $ multiplications of an element in $ \mathbb{F}_{q^m} $ with an element in $ \mathbb{F}_q $. This complexity can be further reduced to about $ \mathcal{O}(n(n-\rho)/\ell) $ such multiplications by the block diagonal form of $ A $.

\subsection{Including Wire-tapper Observations for Coherent Communication} \label{subsec decoding for coherent with wire-tapper}

Consider a nested coset coding scheme, as in Definition~\ref{definition NLCP}, using linearized Reed-Solomon codes $ \mathcal{C}_2 \subsetneqq \mathcal{C}_1 \subseteq \mathbb{F}_{q^m}^n $. If $ k_1 = \dim(\mathcal{C}_1) $ and $ k_2 = \dim(\mathcal{C}_2) $, we may choose $ \mathcal{W} $ in Definition~\ref{definition NLCP} as the vector space generated by the last $ k_1 - k_2 $ rows of the generator matrix of $ \mathcal{C}_1 $ given as in (\ref{eq gen matrix of lin RS code}). 

The encoding is as follows. The message is $ \mathbf{x}_2 \in \mathbb{F}_{q^m}^{k_1 - k_2} $. We generate uniformly at random $ \mathbf{x}_1 \in \mathbb{F}_{q^m}^{k_2} $, and we encode $ \mathbf{x} = (\mathbf{x}_1, \mathbf{x}_2) \in \mathbb{F}_{q^m}^{k_1} $ using the generator matrix of $ \mathcal{C}_1 $ in (\ref{eq gen matrix of lin RS code}) to obtain the codeword $ \mathbf{c} \in \mathbb{F}_{q^m}^n $. 

Finally, the numbers of errors $ t $ and erasures $ \rho $ are constrained by $ 2t + \rho \leq n-k_1 $, by Theorem~\ref{th optimal for coherent comm}. Thus, we may apply the decoding algorithm in Subsection~\ref{subsec quadratic reduction} to the larger code $ \mathcal{C}_1 $ and we recover $ \mathbf{x} = (\mathbf{x}_1, \mathbf{x}_2) \in \mathbb{F}_{q^m}^{k_1} $. Thus we recover the message plus the random keys, which we may simply discard (as usual for instance in the wire-tap channel of type II \cite{ozarow} or in secret sharing \cite{shamir}).

\subsection{The Non-coherent Case} \label{subsec decoding for non-coherent}

As seen in the previous subsection, the addition of the random keys for security against a wire-tapper influences the encoding, but not the decoding. Hence we assume $ \mu = 0 $ wire-tapper observations in this subsection.

We now argue as in \cite[Sec.~4.4]{silva-thesis}. Let $ \mathcal{C} \subseteq \mathbb{F}_{q^m}^n $ be a linearized Reed-Solomon code of dimension $ k $. Assume that we transmit the codeword $ \mathbf{c} = (\mathbf{c}^{(1)}, \mathbf{c}^{(2)}, \ldots, \mathbf{c}^{(\ell)}) \in \mathcal{C} $ using the lifting process as in Definition~\ref{def lifting}. This means that the receiver obtains the matrices
$$ Y_i = \left( \begin{array}{cc}
M_{\mathcal{A}}(\mathbf{c}^{(i)}) \\
\hline
I_{n_i}
\end{array} \right) A_i^T + E_i $$
$$ = \left( \begin{array}{cc}
M_{\mathcal{A}}(\mathbf{c}^{(i)} A_i^T + \mathbf{e}^\prime_i) \\
\widehat{A}^T_i
\end{array} \right) = \left( \begin{array}{cc}
M_{\mathcal{A}}(\mathbf{y}^{(i)}) \\
\widehat{A}^T_i
\end{array} \right), $$
where $ \Rk(E_i) \leq t_i $, and $ A_i \in \mathbb{F}_q^{N_i \times n_i} $, for $ i = 1,2, \ldots, \ell $. So now the receiver knows $ \widehat{A} = \diag(\widehat{A}_1, $ $ \widehat{A}_2, $ $ \ldots, $ $ \widehat{A}_\ell) $ instead of $ A $, and we may compute the linearized Reed-Solomon code $ \mathcal{C}^\prime = \mathcal{C} \widehat{A}^T $ as in Subsection~\ref{subsec decoding for coherent}. We may also assume that $ \Rk(Y_i) = N_i $, for $ i = 1,2, \ldots, \ell $, and $ \widehat{A} $ also has full rank.

By Theorem~\ref{th correction capability injection} and Proposition~\ref{prop relation subspace and sum-rank distance}, if $ 2t + \rho < n-k $, then a minimum sum-injection distance decoder would give us the message. However, summing in $ i $ Equation (4.54) in \cite[Sec.~4.4]{silva-thesis}, we have that
$$ \dd_{SI}(\Col_{\Sigma}(\mathcal{I}_{\Sigma}(\mathbf{c})), \Col_{\Sigma}(Y)) $$
$$ = \wt_{SR}(\mathbf{y} - \mathbf{c} \widehat{A}^T) + \sum_{i=1}^\ell [ n_i - N_i ]^+, $$
where $ \mathbf{y} = (\mathbf{y}^{(1)}, \mathbf{y}^{(2)}, \ldots, \mathbf{y}^{(\ell)}) $. Since the sum of $ [ n_i - N_i ]^+ $ does not depend on $ \mathbf{c} $, we may instead use a minimum sum-rank distance decoder. Therefore, we have reduced the problem to finding the message corresponding to $ \mathbf{c}^\prime = \mathbf{c} \widehat{A}^T \in \mathcal{C}^\prime = \mathcal{C} \widehat{A}^T $ that minimizes $ \wt_{SR}(\mathbf{y} - \mathbf{c}^\prime) $, and we can find a solution using the algorithm in Subsection~\ref{subsec quadratic reduction} for $ \mathcal{C}^\prime $.

\subsection{Iterative Encoding and Decoding} \label{subsec starting decoding beginning}

We now remark that encoding and decoding linearized Reed-Solomon codes can be done in an iterative manner. 

First, by the block-wise structure of the generator matrix (\ref{eq gen matrix of lin RS code}), we may send the packets corresponding to the $ i $th shot before computing those corresponding to the $ (i+1) $th shot. Second, note that Newton interpolation (Appendix~\ref{app newton interpolation}), the computation of the skew polynomials $ L $ and $ Q $, and the translation to sum-rank decoding as in (\ref{eq transforming basis to P-basis}) and (\ref{eq transforming received words}) are all of iterative nature. Therefore we may perform Steps 1 and 2 in the algorithm iteratively, first in $ i $ then in $ j $, as we receive the packets $ y^{(i)}_j \in \mathbb{F}_{q^m} $. 


\section{Conclusion and Open Problems} \label{sec conclusion}

In this work, we have proposed the use of linearized Reed-Solomon codes
for reliability and security in multishot random linear network coding
under a worst-case adversarial model. We have shown that the
corresponding coding schemes achieve the maximum secret message size in
the coherent case, and close to maximum information rate in the
non-coherent case. Moreover, the encoding and decoding can be performed
iteratively and with overall quadratic complexity. Their advantage with
respect to simply using Gabidulin codes is that their field size is roughly
$ \max \{ \ell, q_0 \}^{n^\prime} $ (polynomial in $ \ell $), in contrast to 
$ q_0^{\ell n^\prime} $ (exponential in $ \ell $), where $ n^\prime $ 
is the number of outgoing links at the source, $ \ell $ is the number of shots,
and $ q_0 $ is the base field size of the underlying linear network code. 
This is translated into a reduction of more than one degree in $ \ell $ in the
encoding and decoding complexity in number of operations over $ \mathbb{F}_2 $ (Table \ref{table comparison}).

We now list some open problems for future research:

1) The number of shots is restricted to $ \ell < q $ when using
linearized Reed-Solomon codes. In the Hamming-metric case ($ n^\prime =
1 $), it is well-known that an MDS code satisfies $ \ell < 2q $
\cite{singleton}. Since MSRD codes are MDS, we deduce that $ \ell < 2
q^m / n^\prime $. We conjecture, but leave as an open problem, that MSRD
codes must satisfy $ \ell = \mathcal{O}(q) $ instead of $ \ell =
\mathcal{O}(q^m / n^\prime) = \mathcal{O}(q^m) $, for fixed values of $
m $ and $ n^\prime $.  

2) Faster decoding algorithms exist for Gabidulin codes
\cite{subquadratic, silva-fast}, corresponding to the case $ \ell = 1 $.
We leave as an open problem finding analogous reductions of the decoding
complexity of linearized Reed-Solomon codes.

3) Although unique decoding works analogously for all linearized
Reed-Solomon codes, list-decoding seems to differ in the Hamming-metric
($ n^\prime = 1 $) and rank-metric ($ \ell = 1 $) cases (see
\cite{antonia-list-decoding, list-decodable-rank-metric} and the
references therein). This suggests, but we leave as an open problem, that
the list-decoding of sum-rank codes and linearized Reed-Solomon codes
behaves differently with respect to $ n^\prime $ and $ \ell $.

4) Both the rank-metric list-decodable codes from
\cite{list-decodable-rank-metric} and most maximum rank distance codes
when $ n^\prime > m $  are only linear over $ \mathbb{F}_q $, instead of
$ \mathbb{F}_{q^m} $. We leave as an open problem the study of the security
performance of $ \mathbb{F}_q $-linear codes in multishot network
coding, as done in \cite{rgmw}.

5) Linearized Reed-Solomon codes require using the same base field $ q $
and packet length $ m $ in every shot of the network. We leave as an open
problem the construction of MSRD codes for different base fields and
packet lengths in different shots.

6) It seems natural for future research to study convolutional
sum-rank-metric codes rather than convolutional rank-metric codes. That
is, convolutional codes where we consider an $ \ell $-shot sum-rank
metric in each block rather than the rank metric.  We conjecture that
this should be similar to going from convolutional codes where each
block is simply a field symbol to classical Hamming-metric convolutional
codes.

\appendices

\section{Proofs for Section~\ref{sec measures raliable and secure}} \label{app proofs for section on measures}

In this appendix, we give the proofs of the main results in
Section~\ref{sec measures raliable and secure}. We start with the proof
of Theorem~\ref{theorem correction capability} in Subsection~\ref{subsec
measuring coherent}.

\begin{proof}[Proof of Theorem~\ref{theorem correction capability}]
First assume that the scheme is not $ t $-error and $ \rho $-erasure-correcting. Then there exist integers $ 0 \leq \rho_i \leq n_i $ and full-rank matrices $ A_i \in \mathbb{F}_q^{(n_i - \rho_i) \times n_i} $, for $ i = 1,2, \ldots, \ell $, and there exist vectors $ \mathbf{e}_1, \mathbf{e}_2 \in \mathbb{F}_{q^m}^{n- \rho} $, $ \mathbf{c}_1 \in \mathcal{X}_S $ and $ \mathbf{c}_2 \in \mathcal{X}_T $, where $ S \neq T $, such that
$$ \mathbf{c}_1 A^T + \mathbf{e}_1 = \mathbf{c}_2 A^T + \mathbf{e}_2, $$
where $ \rho = \rho_1 + \rho_2 + \cdots + \rho_\ell $, $ A = \diag(A_1, A_2, \ldots, A_\ell) \in \mathbb{F}_q^{(n - \rho) \times n} $, and $ \wt_{SR}(\mathbf{e}_1), \wt_{SR}(\mathbf{e}_2) \leq t $. By defining $ \mathbf{c} = \mathbf{c}_1 - \mathbf{c}_2 $, we see that 
$$ \wt_{SR}(\mathbf{c} A^T) = \wt_{SR}(\mathbf{e}_2 - \mathbf{e}_1) \leq 2t. $$
Take now full-rank matrices $ B_i \in \mathbb{F}_q^{\rho_i \times n_i} $ such that $ \left( \begin{array}{c}
A_i \\ 
B_i
\end{array} \right) \in \mathbb{F}_q^{n_i \times n_i} $ is invertible, for $ i = 1,2, \ldots, \ell $. If $ B = \diag(B_1, B_2, \ldots, B_\ell) \in \mathbb{F}_q^{\rho \times n} $ and
$$ D = \diag \left( \left( \begin{array}{c}
A_1 \\ 
B_1
\end{array} \right), \left( \begin{array}{c}
A_2 \\ 
B_2
\end{array} \right), \ldots, \left( \begin{array}{c}
A_\ell \\ 
B_\ell
\end{array} \right) \right) \in \mathbb{F}_q^{n \times n}, $$
then $ D $ is invertible and we conclude that
$$ \dd_{SR}(\mathbf{c}_1, \mathbf{c}_2) = \wt_{SR}(\mathbf{c}) = \wt_{SR}(\mathbf{c} D^T) \leq $$
$$ \wt_{SR}(\mathbf{c} A^T) + \wt_{SR}(\mathbf{c} B^T) \leq 2t + \rho, $$
thus $ 2t + \rho \geq \dd_{SR}({\rm Supp}(F)) $.

Assume now that $ 2t + \rho \geq \dd_{SR}({\rm Supp}(F)) $. Take $ \mathbf{c} = \mathbf{c}_1 - \mathbf{c}_2 $ such that $ \mathbf{c}_1 \in \mathcal{X}_S $, $ \mathbf{c}_2 \in \mathcal{X}_T $, $ S \neq T $, and $ \wt_{SR}(\mathbf{c}) \leq 2t + \rho $. There exist full-rank matrices $ B_i \in \mathbb{F}_q^{(2t_i + \rho_i) \times n_i} $ and vectors $ \mathbf{x}_i \in \mathbb{F}_{q^m}^{2t_i + \rho_i} $ such that $ \mathbf{c}^{(i)} = \mathbf{x}_i B_i $, for $ i = 1,2, \ldots, \ell $, where $ t = t_1 + t_2 + \cdots + t_\ell $ and $ \rho = \rho_1 + \rho_2 + \cdots + \rho_\ell $.

Now let $ \mathbf{y}_i \in \mathbb{F}_{q^m}^{2t_i} $ be the first $ 2t_i $ components of $ \mathbf{x}_i \in \mathbb{F}_{q^m}^{2t_i + \rho_i} $, and let $ A_i \in \mathbb{F}_q^{(n_i - \rho_i) \times n_i} $ be a parity-check matrix of the $ \mathbb{F}_q $-linear vector space generated by the last $ \rho_i $ rows in $ B_i \in \mathbb{F}_q^{(2t_i + \rho_i) \times n_i} $, for $ i = 1,2, \ldots, \ell $. Thus, it holds that 
$$ \mathbf{c}^{(i)}A_i^T = \mathbf{x}_i B_i A_i^T = (\mathbf{y}_i, \mathbf{0}) B_i A_i^T, $$
where $ \Rk(M_{\mathcal{A}}(\mathbf{y}_i, \mathbf{0}) B_i A_i^T) \leq \Rk(M_{\mathcal{A}}(\mathbf{y}_i)) \leq 2t_i $, for $ i = 1,2, \ldots, \ell $, and we conclude that $ \wt_{SR}(\mathbf{c} A^T) \leq 2t $, where $ A = \diag(A_1, A_2, \ldots, A_\ell) $. Hence there exists $ \mathbf{e}_1, \mathbf{e}_2 \in \mathbb{F}_{q^m}^{n-\rho} $ with $ \wt_{SR}(\mathbf{e}_1) $, $ \wt_{SR}(\mathbf{e}_1) \leq t $, such that $ \mathbf{c}_1 A^T + \mathbf{e}_1 = \mathbf{c}_2 A^T + \mathbf{e}_2 $. Therefore the codewords $ \mathbf{c}_1 $ and $ \mathbf{c}_2 $ correspond to distinct secret messages, but cannot be distinguished by any decoder. Thus the scheme is not $ t $-error and $ \rho $-erasure-correcting.
\end{proof}

We now prove Theorem~\ref{th correction capability injection} in Subsection~\ref{subsec measuring non-coherent}. We will make repeated use of the following lemma, which is a particular case of \cite[Lemma~15]{on-metrics}.

\begin{lemma} [\textbf{\cite{on-metrics}}] \label{lemma 15 in onmetrics}
For matrices $ X, Y \in \mathbb{F}_q^{m \times n} $, and integers $ \rho \geq 0 $ and $ N \geq n - \rho > 0 $, it holds that
\begin{equation*}
\begin{split}
\min \{ & \Rk(X A^T - Y B^T) \mid \\
& A,B \in \mathbb{F}_q^{N \times n}, \Rk(A), \Rk(B) \geq n - \rho \} \\
 = & [\max \{ \Rk(X), \Rk(Y) \} \\
 & - \dim(\Col (X) \cap \Col (Y) ) - \rho]^+,
\end{split}
\end{equation*}
where $ a^+ = \max \{ 0,a \} $, for $ a \in \mathbb{Z} $. 
\end{lemma}

\begin{proof}[Proof of Theorem~\ref{th correction capability injection}]
First assume that the scheme is not $ t $-error and $ \rho $-erasure-correcting. Then there exist integers $ 0 \leq \rho_i \leq n_i $ and matrices $ A_i, B_i \in \mathbb{F}_q^{N_i \times n_i} $ with $ \Rk(A_i), $ $ \Rk(B_i) \geq n_i - \rho_i $, for $ i = 1,2, \ldots, \ell $, and there exist vectors $ \mathbf{e}_1, \mathbf{e}_2 \in \mathbb{F}_{q^m}^N $, $ \mathbf{c}_1 \in \mathcal{X}_S $ and $ \mathbf{c}_2 \in \mathcal{X}_T $, where $ S \neq T $, such that
$$ \mathbf{c}_1 A^T + \mathbf{e}_1 = \mathbf{c}_2 B^T + \mathbf{e}_2, $$
where $ \rho = \rho_1 + \rho_2 + \cdots + \rho_\ell $, $ A = \diag(A_1, A_2, \ldots, A_\ell) $, $ B = \diag(B_1, B_2, \ldots, B_\ell) \in \mathbb{F}_q^{N \times n} $, and $ \wt_{SR}(\mathbf{e}_1) $, $ \wt_{SR}(\mathbf{e}_2) \leq t $. We have that
$$ \wt_{SR}(\mathbf{c}_1 A^T - \mathbf{c}_2 B^T) = \wt_{SR}(\mathbf{e}_2 - \mathbf{e}_1) \leq 2t. $$
By Lemma~\ref{lemma 15 in onmetrics}, we also have that
\begin{equation*}
\begin{split}
\wt_{SR}(\mathbf{c}_1 A^T - \mathbf{c}_2 B^T) \geq & \sum_{i=1}^\ell \left[ \max \{ \wt_R(\mathbf{c}_1^{(i)}), \wt_R(\mathbf{c}_2^{(i)}) \} \right. \\
- & \left. \dim(\Col(\mathbf{c}_1^{(i)}) \cap \Col(\mathbf{c}_2^{(i)})) - \rho_i \right]^+ \\
\geq & \dd_{SI}(\Col_{\Sigma}(\mathbf{c}_1 ) , \Col_{\Sigma}(\mathbf{c}_2)) - \rho,
\end{split}
\end{equation*}
and we conclude that $ 2t + \rho \geq \dd_{SI}( \Col_{\Sigma}({\rm Supp}(F))) $.

Assume now that $ 2t + \rho \geq \dd_{SI}(\Col_{\Sigma} ({\rm Supp}(F))) $. Take $ \mathbf{c}_1 \in \mathcal{X}_S $, $ \mathbf{c}_2 \in \mathcal{X}_T $, $ S \neq T $, such that $ \dd_{SI}(\Col_{\Sigma}( \mathbf{c}_1 ) , \Col_{\Sigma}( \mathbf{c}_2 )) \leq 2t + \rho $. Define now $ m_i = \max \{ \wt_R(\mathbf{c}_1^{(i)}), \wt_R(\mathbf{c}_2^{(i)}) \} - \dim(\Col( \mathbf{c}_1^{(i)} ) \cap \Col( \mathbf{c}_2^{(i)} )) $, for $ i = 1,2, \ldots, \ell $, and
$$ \delta = \left[ \sum_{i=1}^\ell m_i - \rho \right]^+ - \left( \sum_{i=1}^\ell m_i - \rho \right) \geq 0. $$
Next let $ 0 \leq \rho_i \leq m_i $, for $ i = 1,2, \ldots, \ell $, be such that $ \rho = \sum_{i=1}^\ell \rho_i + \delta $, which implies that
$$ \sum_{i=1}^\ell (m_i - \rho_i) = \left[ \sum_{i=1}^\ell m_i - \rho \right]^+. $$
By Lemma~\ref{lemma 15 in onmetrics}, there exist matrices $ A_i, B_i \in \mathbb{F}_q^{N_i \times n_i} $ such that $ \Rk(A_i) $, $\Rk(B_i) \geq n_i - \rho_i $, for $ i = 1,2, \ldots, \ell $, and
$$ \wt_{SR}( \mathbf{c}_1 A^T - \mathbf{c}_2 B^T) = \sum_{i=1}^\ell (m_i - \rho_i) = \left[ \sum_{i=1}^\ell m_i - \rho \right]^+ $$
$$ = \left[ \dd_{SI}(\Col_{\Sigma}(\mathbf{c}_1) , \Col_{\Sigma}(\mathbf{c}_2)) - \rho \right]^+ \leq 2t, $$
where $ A = \diag(A_1, A_2, \ldots, A_\ell) $ and $ B = \diag(B_1, $ $ B_2, $ $ \ldots, $ $ B_\ell) $. As in the previous proof, the codewords $ \mathbf{c}_1 $ and $ \mathbf{c}_2 $ correspond to distinct secret messages, but cannot be distinguished by any decoder. Thus the scheme is not $ t $-error and $ \rho $-erasure-correcting.
\end{proof}

We conclude this appendix with the proof of Lemma~\ref{lemma information leakage calculation}.

\begin{proof}[Proof of Lemma~\ref{lemma information leakage calculation}]
Define the map $ f : \mathbb{F}_{q^m}^n \longrightarrow \mathbb{F}_{q^m}^\mu $ given by
$$ f(\mathbf{c}) = \mathbf{c} B^T, $$
for $ \mathbf{c} \in \mathbb{F}_{q^m}^n $. Observe that $ f $ is a linear map over $ \mathbb{F}_{q^m} $. For the random variable $ X = F(S) $, it follows that
$$ H(X B^T) = H(f(X)) \leq \log_{q^m} | f(\mathcal{C}_1) | $$ 
$$ = \dim(f(\mathcal{C}_1)) = \dim(\mathcal{C}_1) - \dim(\ker (f) \cap \mathcal{C}_1), $$
where the last equality is the well-known first isomorphism theorem. On the other hand, we may similarly compute the conditional entropy as follows. Recall that $ X $ is uniform given $ S $. We have that
$$ H(X B^T \mid S) = H(f(X) \mid S) = \log_{q^m} | f(\mathcal{C}_2) | $$  
$$ = \dim(f(\mathcal{C}_2)) = \dim(\mathcal{C}_2) - \dim(\ker (f) \cap \mathcal{C}_2). $$
We leave for the reader to prove as an exercise that $ \ker(f) = \mathcal{V}_{(\boldsymbol{\mathcal{L}}^\perp)} = \mathcal{V}_{\boldsymbol{\mathcal{L}}}^\perp $, where $ \boldsymbol{\mathcal{L}}^\perp = (\mathcal{L}_1^\perp, \mathcal{L}_2^\perp, \ldots, \mathcal{L}_\ell^\perp) $ (use that $ B_i $ is a parity-check matrix of $ \mathcal{L}_i^\perp $, for $ i = 1,2, \ldots, \ell $). Therefore
$$ I(S; X B^T) = H(X B^T) - H(X B^T \mid S) $$
$$ \leq (\dim(\mathcal{C}_1) - \dim(\mathcal{V}_{\boldsymbol{\mathcal{L}}}^\perp \cap \mathcal{C}_1)) - (\dim(\mathcal{C}_2) - \dim(\mathcal{V}_{\boldsymbol{\mathcal{L}}}^\perp \cap \mathcal{C}_2)) $$
$$ = (\dim(\mathcal{V}_{\boldsymbol{\mathcal{L}}}) - \dim(\mathcal{C}_1^\perp \cap \mathcal{V}_{\boldsymbol{\mathcal{L}}})) - (\dim(\mathcal{V}_{\boldsymbol{\mathcal{L}}}) - \dim(\mathcal{C}_2^\perp \cap \mathcal{V}_{\boldsymbol{\mathcal{L}}})) $$
$$ = \dim(\mathcal{C}_2^\perp \cap \mathcal{V}_{\boldsymbol{\mathcal{L}}}) - \dim(\mathcal{C}_1^\perp \cap \mathcal{V}_{\boldsymbol{\mathcal{L}}}), $$
where the first equality follows from the dimensions formulas for duals, sums and intersections of vector spaces. 

Finally, if $ S $ is uniform in $ \mathcal{S} $, then all inequalities in this proof are equalities.
\end{proof}

\section{Newton Interpolation for Skew Polynomials} \label{app newton interpolation}

In this appendix, we show how to find the skew polynomials $ F_k $ and $ G $ necessary to initialize the algorithm in Subsection \ref{subsec quadratic reduction}. We use the idea behind Newton's interpolation algorithm for conventional polynomials. Newton interpolation for skew polynomials was investigated in \cite{zhang}. An algorithm with quadratic complexity based on K{\"o}tter interpolation was given in \cite[Sec. 4]{skew-interpolation}. The algorithm in this appendix is analogous, but with notation in accordance with the rest of this paper.

Let $ \mathcal{B}_k = \{ b_1, b_2, \ldots, b_k \} \subseteq \mathbb{F}_{q^m} $ be a P-independent set and let $ a_1, a_2, \ldots, a_k \in \mathbb{F}_{q^m} $. By definition of P-independence, the sets $ \mathcal{B}_i = \{ b_1, b_2, \ldots, b_i \} $ are also P-independent, for $ i = $ $1, $ $2, $ $ \ldots, k $. It also holds that $ F_{\mathcal{B}_i}(b_{i+1}) \neq 0 $, for $ i = $ $ 1, $ $2, $ $ \ldots, k-1 $, by definition. 

We have that $ F_{\mathcal{B}_1} = x-b_1 $, and $ G_1 = a_1 $ satisfies that $ \deg(G_1) < 1 $ and $ G_1(b_1) = a_1 $. Now, let $ 1 \leq i \leq k-1 $ and assume that we have computed $ F_{\mathcal{B}_i} $ and $ G_i $ such that $ \deg(G_i) < i $ and $ G_i(b_j) = a_j $, for $ j = 1,2, \ldots, i $. Then by Lemma \ref{lemma product rule}, it holds that
\begin{equation*}
\begin{split}
F_{\mathcal{B}_{i+1}} = & (x - b_{i+1}^{F_{\mathcal{B}_i}(b_{i+1})}) F_{\mathcal{B}_i}, \\
G_{i+1} = & G_i + (a_{i+1} - G_i(b_{i+1})) F_{\mathcal{B}_i}(b_{i+1})^{-1} F_{\mathcal{B}_i},
\end{split}
\end{equation*}
where $ G_{i+1} $ is the unique skew polynomial with $ \deg(G_{i+1}) < i+1 $ and $ G_{i+1}(b_j) = a_j $, for $ j = 1,2, \ldots, i+1 $. Thus the required skew polynomials are $ F_k = F_{\mathcal{B}_k} $ and $ G = G_k $.

This algorithm requires computing the evaluations $ F_{\mathcal{B}_i}(b_{i+1}) $ and $ G_i(b_{i+1}) $, where $ \deg(F_{\mathcal{B}_i}), \deg(G_i) \leq i < k $, which requires $ \mathcal{O}(k) $ multiplications and $ \mathcal{O}(k) $ additions. It requires another $ \mathcal{O}(k) $ multiplications and $ \mathcal{O}(k) $ additions in each of the $ k-1 $ steps. Thus its complexity is of $ \mathcal{O}(k^2) $ multiplications and $ \mathcal{O}(k^2) $ additions in $ \mathbb{F}_{q^m} $.

\ifCLASSOPTIONcaptionsoff
  \newpage
\fi



\bibliographystyle{IEEEtran}
%



%





\end{document}